\documentclass{wlscirep}
\usepackage[utf8]{inputenc}
\usepackage[T1]{fontenc}

\usepackage{amsmath,amssymb,xspace,epsfig}

\usepackage{color}
\usepackage{graphicx}
\usepackage{hyperref}
\usepackage[numbers,sort]{natbib}

\usepackage{mystyle}
\usepackage[linesnumbered,ruled]{algorithm2e}
\usepackage{wrapfig,subfigure}
\let\eps\varepsilon

\long\def\commented#1{}
\def\etal{\emph{et~al.}\xspace}

\newcommand{\denselist}{\itemsep 0pt\parsep=1pt\partopsep 0pt}
\newcommand{\bitem}{\begin{itemize}\denselist}
\newcommand{\eitem}{\end{itemize}}
\newcommand{\benum}{\begin{enumerate}\denselist}
\newcommand{\eenum}{\end{enumerate}}

\long\def\commented#1{}
\def\etal{\emph{et~al.}\xspace}
\let\eps\varepsilon

\def\R{\mathbb{R}}

\let\eps\varepsilon

\def\etal{\emph{et~al.}\xspace}

\title{Community Detection on Networks\\ with Ricci Flow}

\author[1]{Chien-Chun Ni}
\author[2]{Yu-Yao Lin}
\author[3]{Feng Luo}
\author[4,*]{Jie Gao}
\affil[1]{Yahoo! Research, Sunnyvale, CA, USA}
\affil[2]{Intel Corp., Hillsboro, OR, USA}
\affil[3]{Rugters University, New Brunswick, NJ, USA}
\affil[4]{Stony Brook University, Stony Brook, NY, USA}

\affil[*]{jgao@cs.stonybrook.edu}


\begin{abstract}
	
Many complex networks in the real world have community structures -- groups of well-connected nodes with important functional roles. 
It has been well recognized that the identification of communities bears numerous practical applications. 
While existing approaches mainly apply statistical or graph theoretical/combinatorial methods for 
community detection, in this paper, we present a novel geometric approach which enables us to borrow powerful classical geometric methods and properties. 
By considering networks as geometric objects and communities in a network as a geometric decomposition, we apply curvature and discrete Ricci flow, which have been used to decompose smooth manifolds with astonishing successes in mathematics, to break down communities in networks.
We tested our method on networks with ground-truth community structures, and experimentally confirmed the effectiveness of this geometric approach.

\end{abstract}

\begin{document}

\flushbottom
\maketitle

\thispagestyle{empty}

\section{Introduction}

Complex networks have been used to model connections of elements in many different fields such as social networks, biology, and biochemistry (protein-protein networks~\cite{bhowmick2015clustering}, metabolic networks, and gene networks), and computer science (P2P, the Internet). It has been widely recognized that many of real world networks have community structures -- nodes in the same community are densely connected while nodes from different communities are sparsely connected. Recognition of community structures brings out important functional components and plays an important role in supporting processes on networks such as contagions of diseases, information or behaviors. Many algorithms have been developed to identify and separate communities in the literature~\cite{yang2016comparative,Fortunato2010-pe,Newman2006-xe,Sinha2018-nx,leskovec2010empirical,Clauset2004-bp,Zhang2014-xe,Peel2017-nx,Allen2017-ah,JMLR:v18:16-480}.
%such as label propagation\cite{}, Infomap\cite{}, Fast greedy\cite{} and spinglass\cite{} algorithms. See ~\cite{yang2016comparative,Fortunato2010-pe} for a detailed survey. 

Most current works on community detection try to recognize dense clusters in a graph: by randomized algorithms such as label propagation~\cite{Raghavan2007-wf} or random walks~\cite{Rosvall2008-ns}; by optimized centrality such as betweenness centrality~\cite{Girvan2002-cq}; or by considering notions such as modularity~\cite{newman2004finding}: the fraction of edges that fall within the given groups minus the expected fraction if edges were distributed uniformly at random (while still respecting the degree distribution). The viewpoint of modularity could be considered as a statistical measure of non-uniformity of the network. 
% The fundamental difference in this work is to take a geometric viewpoint, and consider geometric measures that more fundamentally capture the inherent differences within and across communities. We borrow the idea of Ricci flow from manifold deformation to graph decomposition.

Unlike existing methods, our work explores a new path connecting community detection and geometry. We consider community structure as a geometric phenomenon and use geometric methods to identify communities in a network. The motivation comes from the classical topological connected sum decomposition of $3$-manifolds. The groundbreaking work of Hamilton and Perelman~\cite{hamilton1982three, Perelman2002-gs} shows that the connected sum decomposition can be detected by the geometric Ricci flow. By considering a network as a discrete counterpart of a manifold and connected sum components as communities, we introduce a discrete Ricci flow on networks for identifying communities in a network.

%\saiba{please check the two paragraphs below}

The Ricci flow approach is based on the geometric notion of curvature, introduced by F. Gauss and B. Riemann over 150 years ago, which describes quantitatively how spaces are bent at each point~\cite{Jost2017-bx}. In classical geometry, regions in a space with large positive curvature tend to be more densely packed than regions of negative curvature. 
To locate these regions of large curvature, in a seminal work in 1982, Hamilton~\cite{hamilton1982three} introduced a curvature guided diffusion process, called the Ricci flow, that deforms the space in a way formally analogous to the diffusion of heat. Under the Ricci flow, regions in a space of large positive curvature shrink to points whereas regions of very negative curvature spread out.
% The groundbreaking work of Hamilton and Perelmann % ~\cite{mt,perelman1, perelman2, hamilton} 
%on Ricci flow shows that the flow can be used to detect topological decompositions of 3-dimensional manifolds. 
In this paper, we observed that communities in networks resemble regions in Riemannian manifolds of large positive curvature. By applying the discrete Ricci flow on networks as the classic Ricci flow on manifolds, we are able to detect community structures in networks. 

\begin{figure}[htbp]
    \centering
    \includegraphics[width=0.5\columnwidth]{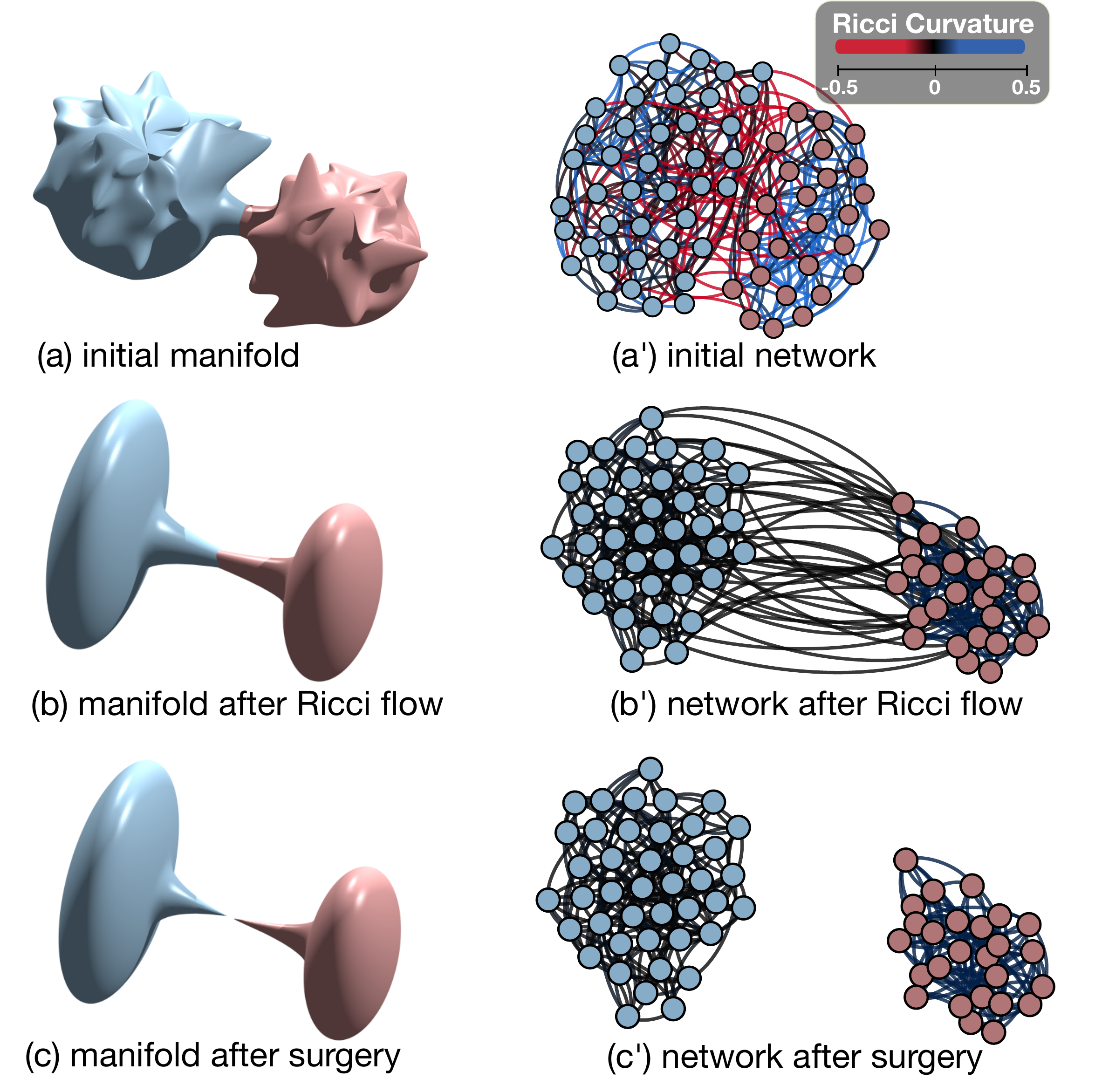}
    \caption{An illustration of Ricci flow on a manifold and a network. 
	% By resembling communities to large positive curvature regions in Riemannian manifolds and networks to the discretion of manifolds,
	Ricci flow captures the large positive curvature regions in the manifold as well as communities in the network. The formation of singularities in the Ricci flow is illustrated by (b, c).  The analog in (b', c') in the discrete Ricci flow decomposes the network into communities.}
    \label{fig:manifold} %% label for entire figure
\end{figure}

Figure~\ref{fig:manifold} illustrates this key observation. In the left column, the Ricci flow deforms a Riemannian manifold(Figure~\ref{fig:manifold}(a)) gradually and develops a neck pinching singularity (Figure~\ref{fig:manifold}(b)). By removing the singularity, the manifold is decomposed into sub-regions of positive curvature(Figure~\ref{fig:manifold}(c)). In the right column, the discrete Ricci flow on a metric graph(Figure~\ref{fig:manifold}(a')) stretches edges of large negative Ricci curvature and shrinks edges of large positive Ricci curvature over time(Figure~\ref{fig:manifold}(b')). By removing the edges of weight greater than a threshold value, we recover subgraphs of large Ricci curvature representing communities(Figure~\ref{fig:manifold}(c')).

\subsection{Our Contribution}

To carry out the discretization process, we start from the recent important work of Y. Ollivier \cite{Ollivier2009-yl,Ollivier2010-ub, Lott2009-lp} which introduced Ricci curvature on metric measure spaces by using the optimal transport theory. %with a probability measure at each point. 
%Ollivier derives the classical Ricci curvature on a smooth manifold is closely related to the mathematical theory of optimal transport. Therefore, the theory of optimal transport, which is formulated on general spaces including networks,  can be used to define Ricci curvature on such spaces.  
%The optimal transport problem originated by G. Monge in 1781 is to minimize the transportation cost to move iron ores from different mines to a collection of factories which consume the iron ores. 
Ollivier's definition for metric graphs assigns a probability measure to each node and the Ricci curvature of an edge is related to 
the optimal transportation cost %(at a geodesic segment) 
between two probability measures defined on the vertices of the edge.
%is negatively propositional to the Ricci curvature of this edge. 
%The interpretation of  Ollivier's theorem on networks says an edge of large Ollivier Ricci curvature has small transportation cost. This confirms again the intuition that edges within a community should be sparsely traveled (smaller transportation cost) compared with edges connecting different communities.
%This curvature definition has been used in graph analysis for applications such as anomaly detection, detection of backbone edges or cancer related proteins
Various definitions of Ricci curvatures on networks have been used in graph analysis for applications such as anomaly detection, detection of backbone edges or cancer related proteins~\cite{Ni2015-yv,Samal2018-bt,Sreejith2016-yl,Wang2014-iy,Whidden2017-xa,Jost2014-rk,Sandhu2015-lz,Sandhu2016-vh,Ni2018-pv}.
%\textcolor{red}{Feng: are we sure that Ollivier's Ricci has been used in all of the above works?  I think these papers use different definition of Ricc curvatures, like Forman's, or Emary-Bakery Ricci curvatures. At least Jost work used Froman for sure.  How about modify it: Various defintions of Ricci curvatures on networks have been used in graph analysis for applications such as anomaly detection, detection of backbone edges or cancer related proteins}

Motivated by Hamilton's Ricci flow, we introduce an algorithm, called discrete Ricci flow on networks, for detecting community structures.
%and using the fact that optimal transportation cost on a network is a linear programming problem
The discrete Ricci flow is defined on weighted graphs and deforms edge weights as time progresses: edges of large positive Ricci curvature (i.e., sparsely traveled edges) will shrink and edges of very negative Ricci curvature (i.e., heavily traveled edges) will be stretched. 
By iterating the Ricci flow process, we are able to identify heavily traveled edges and thus find communities. 

Figure~\ref{fig:karate_demo} illustrates how discrete Ricci flow detects communities on the Zachary's Karate club graph. 
In this graph, individuals in the same club are represented as nodes of the same color, and friendship ties between two individuals are represented as edges with weights equal to $1$. 
With discrete Ricci flow algorithm, the edge weights evolve such that the community structure can be easily detected by removing edges that are stretched greater than a threshold. Figure~\ref{fig:fb} shows another example of communities on a Facebook ego graph. We have also tested our Ricci flow algorithms on many of the real-world networks with ground-truth communities and artificial networks, and shown competitive accuracy results with other community detection algorithms using various statistical methods or physics models.

\begin{figure}[htbp]
    \centering
    \includegraphics[width=1\columnwidth]{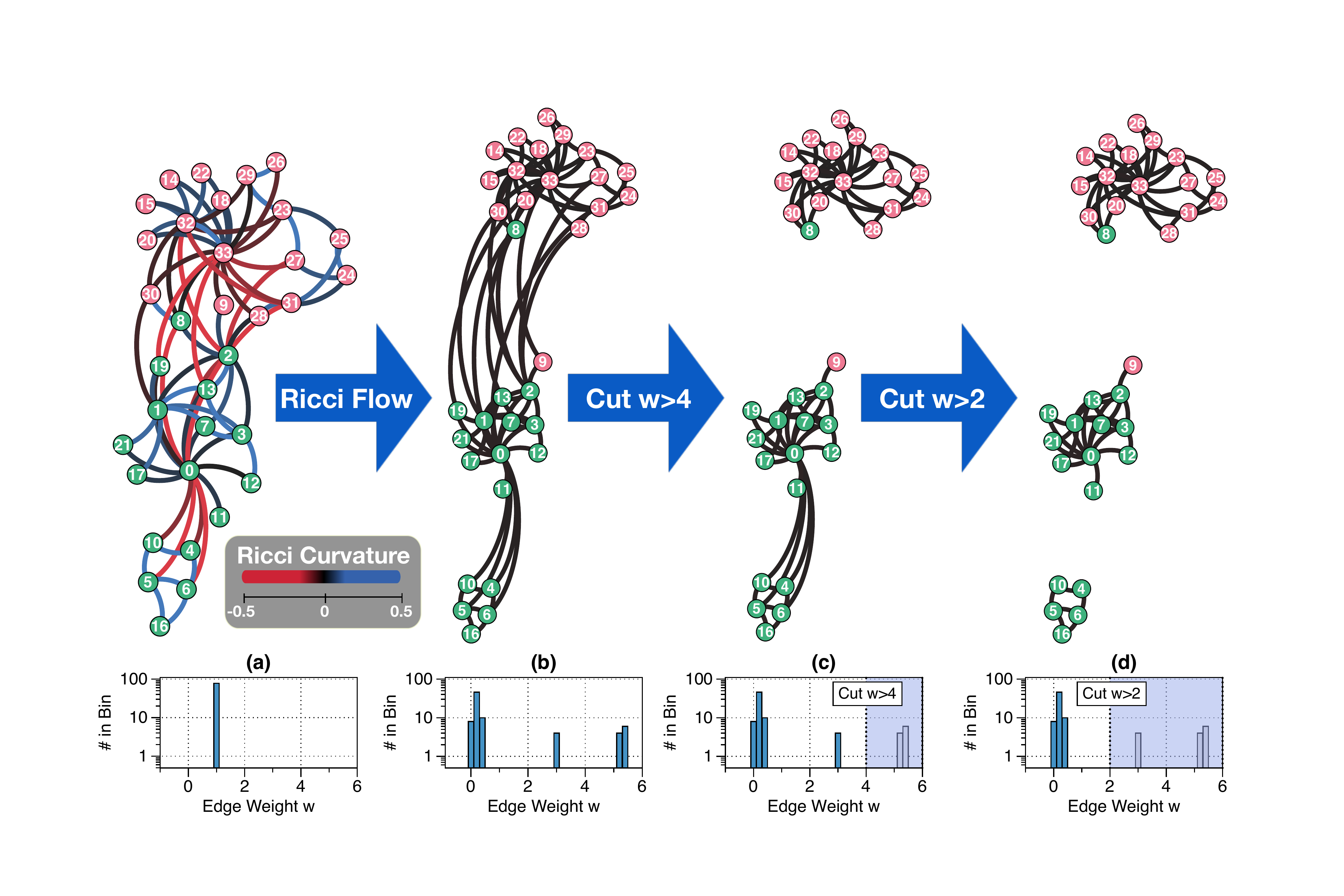}
    \caption{Ricci flow for community detection on the Karate club graph generated by Gephi's ForceAtlas2 layout\cite{ICWSM09154}. (a): The Karate club graph with edge weight $1$ on all edges. Different colors of vertices represent different communities. The colors of edges represent the Ricci curvature on the edges. Notice that most edges between communities are negatively curved. (b): The same graph after $100$ Ricci flow iterations. Ricci flow adjusts the edge weights so that the edge Ricci curvatures are the same everywhere; the intra-community edges shrink; the inter-community edges are stretched. (c) By removing all edges with weight greater than $4$, we acquire a partitioning of the graph with two communities. (d) By removing all edges with weight greater than $2$, we obtain three communities.}
    \label{fig:karate_demo} %% label for entire figure
\end{figure}

\begin{figure}[htbp]
    \centering
    % \vspace{-2mm}
    % \subfigure[]{
    %     \label{fig:fb:subfig:init}
    %     \includegraphics[width=0.31\columnwidth]{figures/1684_init}}
    % \subfigure[]{
    %     \label{fig:fb:subfig:before}
    %     \includegraphics[width=0.31\columnwidth]{figures/1684_before-min}}
    % % \hspace{0.0in}
    % \subfigure[]{
    %     \label{fig:fb:subfig:after}
    %     \includegraphics[width=0.31\columnwidth]{figures/1684_after-min}}    
    % \vspace{-2mm}
    \includegraphics[width=1\columnwidth]{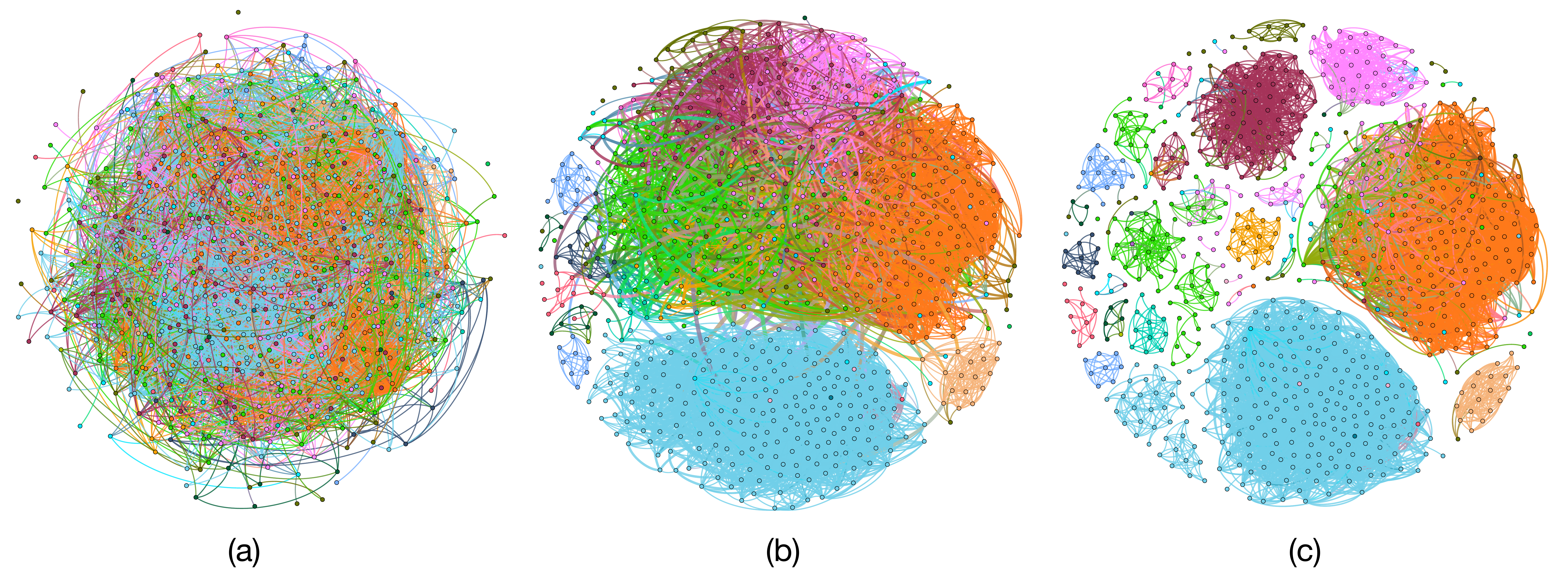}
    \caption{(a) A Facebook ego network of one user with $792$ friends and $14025$ edges generated by Gephi's Fruchterman Reingold layout\cite{ICWSM09154}. The colors represent $24$ different friend circles (communities) hand labeled by the user. (b) By the Ricci flow process of 20 iterations, the weights of inter-community edges are increased (thick edges in the figure) while the weights of intra-community edges gradually shrink to $0$ (thin edges in the figure). (c) By removing the inter-community edges with high weights, the communities are clearly detected.}
    %  \vspace{-4mm}
    \label{fig:fb} %% label for entire figure
\end{figure}

We applied the discrete Ricci flow method on artificial networks generated by the stochastic block model(SBM)~\cite{JMLR:v18:16-480}, the Lancichinetti-Fortunato-Radicchi (LFR) benchmark graph~\cite{lancichinetti2008benchmark} and the emergent geometrical network model~\cite{Bianconi2014-cw,wu2015emergent} (GNet). We choose Adjusted Rand Index (ARI)~\cite{hubert1985comparing} as a quality measure for the clustering accuracy. The proposed discrete Ricci flow method is shown to provide nearly perfect clustering result when community structures exist. Also, extensive comparison tests on real networks with ground-truth communities show that our algorithm is competitive with previously proposed ones. Similar results have been observed with other metrics of clustering accuracy such as modularity.

Our work of Ricci curvature on networks is built on our previous works~\cite{Ni2015-yv,Ni2018-pv} and is also inspired by the important works of E. Saucan and J. Jost \etal in \cite{Samal2018-bt, Saucan2018-yg, Sreejith2017-md, Jost2014-rk}. In these works, they systematically introduced and investigated various discrete curvatures for complex networks. 
The comparative analysis of Forman and Ollivier Ricci curvature on benchmark datasets of complex networks and real-world networks was also carried out. 
Their numerical results show a striking fact that these two completely different discretizations of the Ricci curvatures are highly correlated in many networks.

% Add one paragraph:  We have tested our Ricci flow algorithms on many of the real-world networks and artificial networks ...... And the results are ......

% Add one paragraph:  Our work on Ricci curvature on networks are builds on our previous works and is also inspired by the important work of E. Saucan and J. Jost et al....  There have been many important works on community detection using various satistical methods or physicis models.  .... 

\subsection{Related Work}

Ricci curvature on general spaces without Riemannian structures has been recently studied, in the work of Ollivier~\cite{Ollivier2010-ub, Ollivier2009-yl} on Markov chains, and Bakry and Emery~\cite{Bakry1985-tp}, Lott, Villani~\cite{Lott2009-lp}, Bonciocat and Sturm \cite{Bonciocat2009-tq, Bonciocat2014-qw} on general metric spaces. 
%In the work of  Chung and Yau~\cite{Chung1996-di}, they introduced a Ricci curvature on graphs.
Ricci curvature based on optimal transportation theory, proposed by Ollivier (Ollivier-Ricci curvature)~\cite{Ollivier2010-ub, Ollivier2009-yl}, has become a popular topic and been applied in various fields -- for distinguishing cancer-related genes from normal genes~\cite{Sandhu2015-lz}, for studying financial market fragility~\cite{Sandhu2016-vh}, for understanding phylogenetic trees~\cite{Whidden2017-xa}, and for detecting network backbone and congestion~\cite{Ni2015-yv,Wang2014-iy, Wang2016-pm}. In \cite{Pal2018-cy}, Pal \etal proposed to use Jaccard coefficients for a proxy for Ollivier-Ricci Curvature.
Besides, discrete Ricci curvature has also been defined on cell complexes, proposed by Forman~\cite{Forman2003-ao} (Forman curvature or Forman-Ricci curvature). Forman curvature is based on graph Laplacian. It is easier and faster to compute than Ollivier-Ricci curvature, but is less geometrical. It is more suitable for large scale network analysis~\cite{Samal2018-bt, Weber2017-hu, Sreejith2016-yl,weber2018detecting} and image processing~\cite{Saucan2009-ih}. We have also experimented with Forman curvature for community detection. The results were less satisfying. So here we focus on Ollivier-Ricci curvature.

Unlike discrete Ricci curvature, discrete Ricci flow has not been studied as much. Chow and Luo introduced the first discrete Ricci flow on surfaces~\cite{chow2003combinatorial}. In \cite{Weber2017-hu}, Weber \etal suggested applying Forman-Ricci flow for anomaly detection in the complex network. In \cite{Ni2018-pv}, Ni \etal used the Ollivier-Ricci curvature flow to compute the Ricci flow metric as edge weights for the problem of network alignment (noisy graph matching).

Community detection is a well-studied topic in social network analysis~\cite{yang2016comparative, Plantie2013-dy, leskovec2010empirical, Fortunato2010-pe,Pares2018-yq,Yin2017-nq,Newman2016-nj,Decelle2011-fh}, and protein-protein interaction networks~\cite{Ji2014-cw, bhowmick2015clustering}. There are a few main ideas. One family of algorithms iteratively remove edges of high `centrality', for example, the edge betweenness centrality as suggested in \cite{Girvan2002-cq} by Girvan and Newman. The other idea is to use modularity (introduced by Newman and Clauset \etal~\cite{Newman2006-xe, Clauset2004-bp}), which measures the strength of division of a graph into clusters, as the objective of optimization.
But methods using modularity suffer from a resolution limit and cannot detect small communities.
%is proposed to increase the modularity to find communities in large graphs. 
A geometric extension, named Laplacian modularity, is also suggested with the help of Gauss's law in \cite{Sinha2018-nx}.
Another family of algorithms borrows intuitions from other fields. 
In \cite{Reichardt2006-zs}, a spin glass approach uses the Potts model from statistical physics: every node (particle) is assigned one of $c$ spin states (communities); edges between nodes model the interaction of the particles. 
The community structure of the network is understood as the spin configuration that minimizes the energy of the spin glass.
%In the end, the communities are decided by the non-deterministic spin states after iterations. 
In \cite{Raghavan2007-wf}, Raghavan \etal proposed a non-deterministic label propagation algorithm for large networks. In the initial stage, the algorithm randomly assigns each node in the graph one of $c$ labels. Each node then changes its label to the most popular label among its neighbors. 
%Similar to the spin glass algorithm, the final communities of label propagation algorithm are decided by the final label assignment. 
Infomap~\cite{Rosvall2008-ns} uses an information theoretic approach. A group of nodes for which information flows quickly shall be in the same community. The information flow is approximated by random walks and succinctly summarized by network coding. %which refers to a flow-based map equation in information theory. The map equation finds community structure while minimizing the description length of a random walker's movements on a network. The map equation is defined such that the description can be abbreviated if the network has regions in which the random walker tends to stay for a long time. Therefore, minimizing the description length over all possible network partitions captures community structure with respect to the dynamics on the network. 
%Initially, each node is assigned to its own community. Then, in a random sequential order, each node is moved to the neighboring community that results in the largest decrease of the description length. If no move results in a decrease of the description length, then the node stays in its original community. The procedure is repeated until no movement generates a decrease of the description length.

Taking a geometric view of complex networks is an emerging trend, as shown in a number of recent work. For example, the community structures were used as a coarse version of its embedding in a hidden space with hyperbolic geometry~\cite{Faqeeh2018-lo}. Topological data analysis, a typical geometric approach for data analysis, has been applied for analyzing complex systems~\cite{Salnikov2018-he}.

% \cite{Pares2018-yq}
% Fluid Communities: A Competitive, Scalable and Diverse Community Detection Algorithm community detection

% \cite{Yin2017-nq}
% Local Higher-Order Graph Clustering

%\medskip
%S1.1  Optimal transportation and the work of Ollivier on generalizing Ricci to metric measure space
%S1.2  Our work on discrete Ricci flow on networks
%S1.3  The classical curvatures in differential geometry.
%S1.4  The Ricci curvature and its basic properties
%S1.5  The work on Hamilton and Perelman on Ricci flow on 3-manifolds
%S1.6  Discussion of discrete Ricci curvature on networks and its
%relationship to community struct%S1.7  Mathematical proof of the community detection of graphs G(a,b) using Ricci flow

\section{Classical theory of Ricci curvature, Optimal transport and the Ricci flow}
%\input{scientific_report-math-theory.tex}
%\section{Ricci curvature, optimal transport and %Ollivier's Ricci curvature}
% \medskip

In this section, we briefly recall the basic notation of Ricci curvature in Riemannian geometry, Ollivier's work on generalizing Ricci curvature to metric measure spaces through optimal transport, and the Ricci flow. Their discrete and computational counterparts are addressed in  Section~\ref{sec:main-theory}.

% \medskip
\subsection{Sectional and Ricci curvature} 
One of the central themes in modern geometry is the notion of curvature which quantitatively measures how space is curved. It was introduced by Gauss and Riemann. For a surface in the 3-dimensional Euclidean space, the \emph{Gaussian curvature} at a point is defined as the signed area distortion of the Gauss map sending a point on the surface to its unit normal vector. For instance, a plane has zero curvature, a sphere has positive curvature and a hyperboloid of one sheet has negative curvature (Figure~\ref{fig:curvature}). Gauss showed that curvatures depend only on the induced Riemannian metric on the surface, i.e., independent of how a surface is embedded in the 3-dimensional space. 

\begin{figure}[htbp]
    \centering
    % \subfigure[Surface of negative curvature]{
    %     \label{fig:gaussian:subfig:negative}
    %     \includegraphics[width=0.25\columnwidth]{figures/NegativeCurvature}}
    % \subfigure[Surface of zero curvature]{
    %     \label{fig:gaussian:subfig:zero}
    %     \includegraphics[width=0.25\columnwidth]{figures/ZeroCurvature}}
    % \subfigure[Surface of positive curvature]{
    %     \label{fig:gaussian:subfig:positive}
    %     \includegraphics[width=0.25\columnwidth]{figures/PositiveCurvature}}\\
    % \subfigure[Negative curvature]{
    %     \label{fig:curvature:subfig:negative}
    %     \includegraphics[width=0.18\columnwidth]{figures/neg.pdf}}
    % \subfigure[Zero curvature]{
    %     \label{fig:curvature:subfig:zero}
    %     \includegraphics[width=0.25\columnwidth]{figures/zero2_mod.pdf}}
    % \subfigure[Positive curvature]{
    %     \label{fig:curvature:subfig:positive}
    %     \includegraphics[width=0.18\columnwidth]{figures/pos.pdf}}
    \includegraphics[width=1\columnwidth]{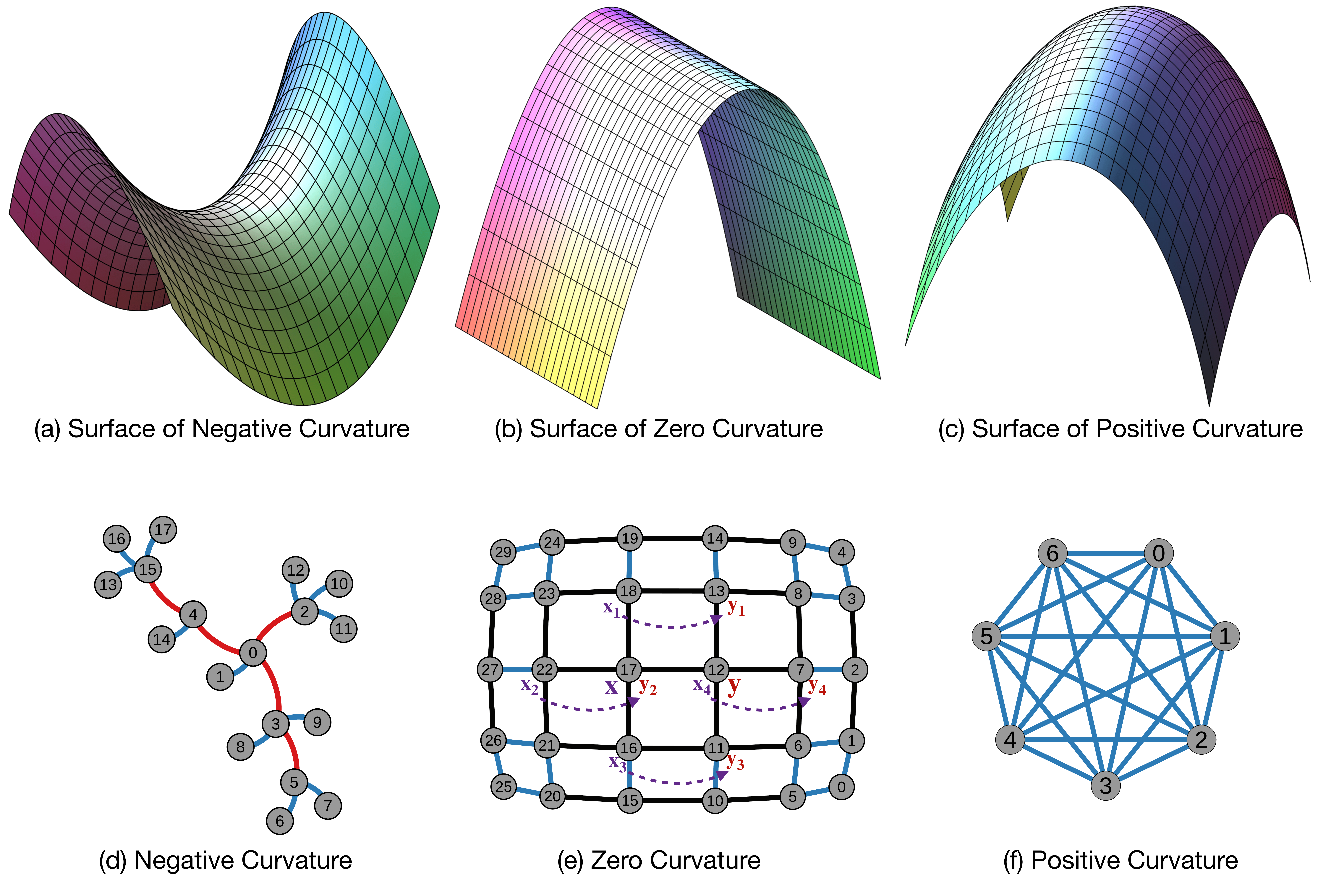}
    \caption{Examples of Ricci curvature on manifolds and graphs. In (a)-(c), manifolds with negative, zero, and positive curvatures are shown. In (d)-(f), all edges have weight of $1$. (d) A tree graph with negative curvature everywhere, except the edges of the leaf nodes. (e) A (infinitely sized) grid graph with all edges of zero curvature. The cost of moving $m_x=\{x, x_1,x_2,x_3,x_4\}$ to $m_y=\{y,y_1,y_2,y_3,y_4\}$ is equal to $d(x,y)$ (f) A complete graph with all edges of positive curvature. }
    \label{fig:curvature} %% label for entire figure
\end{figure}

For a Riemannian manifold $(M, g)$, Riemann's \emph{sectional curvature} assigns a scalar to each 2-dimensional linear subspace $P$ in the tangent space at a point $p$ of $M$. This scalar is equal to the Gaussian curvature of the image of $P$ under the exponential map at $p$.
A positive sectional curvature space tends to have a small diameter and is geometrically crowded (e.g., a sphere). In contrast, a negative sectional curvature closed Riemannian manifold has an infinite fundamental group, a contractible universal cover, and is geometrically spreading out like a tree in large scale. Thus, a positively curved region behaves more like a `community' than negatively curved regions.
%\saiba{please check this sentence=>}
Similar to sectional curvature, the \emph{Ricci curvature} assigns each unit tangent vector $v$ at $p$ a scalar which is the average of the sectional curvatures of
planes containing $v$. Geometrically, Ricci curvature
controls how fast the volume of a ball grows as a function
of the radius. It also controls the volume of the overlap of two balls
in terms of their radii and the distance between their centers.
%\saiba{Add example}
On the other hand, the volume of the overlap of two balls is directly related to the cost of transportation to move one ball to the other, i.e., a larger volume of overlap means less cost of moving one ball to the other. It shows that the Ricci curvature is related to optimal transportation. An explicit formula (Equation~\ref{eq345}) that builds a bridge between them was worked out by Ollivier~\cite{Ollivier2009-yl}. Through the formula, Ollivier defined the generalized Ricci curvature on metric measure spaces by the optimal transportation.
%of the Riemannian volume measure~\cite{Ollivier2009-yl}. 
%This is one of the main tools we use for understanding geometry of the network.

\subsection{The optimal transportation and Ollivier's Ricci curvature.}
% In this work, we have studied the geometry of networks along with the optimal transportation. 
The original optimal transport problem was proposed by G. Monge in 1781. The problem wants to minimize the transportation cost to move iron ores from different mines to a collection of factories which consume the iron ores. In Monge's setting, the problem can be mathematically formulated as follows.
% Let $X$ and $Y$ be two metric spaces; $\mu$ and $\nu$ be two probability Borel measures on $X$ and $Y$ respectively. 
Let mines and factories as two probability spaces $X$ and $Y$; the amount of iron ores to be moved and consumed as two probability Borel measures $\mu$ and $\nu$, we define the cost of transporting from location $x$ to location $y$ to be $c(x,y)$, where $c : X \times Y  \to \R_{\geq 0}$.
% In Monge's setting, $X$ and $Y$ stand for mines and factories and $\mu$ and $\nu$ are the amount of iron ores to be moved and consumed respectively. 
% Let  $c : X \times Y  \to \R_{\geq 0}$
% be a continuous function considered as the cost, i.e., the cost of transporting from location $x$ to location $y$ is $c(x,y)$. 
In general, the cost function $c$ is usually taken to be the distance $d(x, y)$ if $X=Y$ and the cost of transportation per-unit distance is constant. A \emph{transportation} $T: (X, \mu) \to (Y, \nu)$ is a measure
preserving map. Monge's formulation of the optimal transportation
problem is to find a transportation $T : X \to Y$ that realizes ${ \inf \left\{\left.\int
_{X}c(x,T(x))\,\mathrm {d} \mu (x)\;\right| \text{T: transportation} \right\}}$.

Monge's optimal transportation problem had a major breakthrough in 1930 when Kantorovich formulated the optimal transportation problem into a linear optimization problem. In his setting, Kantorovich replaces transportation maps $T$ by probability measures $\gamma$ on $X \times Y$ (called \emph{transportation plans}) satisfying $\gamma(A \times Y)=\mu(A)$
and $\gamma(X \times B)=\nu(B)$ for all measurable subsets $A$ and $B$. The goal is to find a transportation plan $\gamma$
that attains the infimum cost $$W(\mu, \nu) ={\displaystyle \inf
\left\{\left.\int _{X\times Y}c(x,y)\,\mathrm {d} \gamma
(x,y)\right|\gamma \in \Gamma (\mu ,\nu )\right\}},$$ where
$\Gamma(\mu, \nu)$ denotes the collection of all possible transportation plans. If $X$ is a metric space with distance function $d$ and $X=Y$, the quantity $W(\mu, \nu)$ for $c(x,y)=d(x,y)$ is called the \emph{Wasserstein distance} (or the earth mover's distance) between two probability measures $\mu,
\nu$ on $X$.

Wasserstein distance plays a crucial role in Ollivier's approach to Ricci curvature. In his observation~\cite{Ollivier2009-yl},
% Ollivier's approach to Ricci curvature relies on the following key observation~\cite{Ollivier2009-yl}. 
if $(M^n,d)$ is an n-dimensional Riemannian manifold with Riemannian volume $\mu$ and fix $\epsilon>0$, let $m_x=\frac{\mu|_{B(x, \epsilon)}}{\mu(B(x, \epsilon))}$
be the probability measure associated to $x \in M$ where $B(x, \epsilon)$ is the ball of radius $\epsilon$ at $x$. Then the
Wasserstein distance $W(m_x, m_y)=(1-k(x,y))d(x,y)$, where
\begin{equation} \label{eq345} k(x,y)=\frac{\epsilon^2
\text{Ricci}(v,v)}{2(n+2)}+O(\epsilon^3+\epsilon^2 d(x,y))\end{equation} and
$v$ is the tangent vector at $x$ of the geodesic $xy$.  This shows
that Ricci curvature can be defined for general metric spaces with
measures. Given a metric space $(X, d)$ equipped with a
probability measure $m_x$ for each $x\in X$, the Ollivier's Ricci
curvature along the path $xy$ is defined to be 

\begin{equation}\label{eqn:ollivier-ricci}
   \kappa_{xy}=1-\frac{W(m_x, m_y)}{d(x,y)}, 
\end{equation}
where $W(m_x, m_y)$ is the
Wasserstein distance with respect to $c(x, y)=d(x,y)$.

\subsection{The Ricci flow}

The Ricci flow, introduced by Richard S. Hamilton in 1981~\cite{hamilton1982three}, deforms the metric of a Riemannian manifold in a way formally analogous to the diffusion of heat, smoothing out irregularities in the metric. The Ricci flow has been one of the most powerful tools for solving geometric problems in the past forty years. The flow exhibits many similarities with the heat equation.

Suppose a Riemannian metric $g_{ij}$ is given on a manifold $M$ so that its Ricci curvature is $R_{ij}$. Hamilton's Ricci flow is the following second-order nonlinear partial differential equation on symmetric $(0, 2)$-tensors:
$$\frac{\partial}{\partial t}g_{ij}=-2 R_{ij}.$$
A solution to the Ricci flow is a one-parameter family of metrics $g_{ij}(t)$ on a smooth manifold $M$ satisfying the above partial differential equation. One of the key properties of the Ricci flow is that the curvature evolves according to a nonlinear version of the heat equation. Thus the Ricci flow tends to smooth out irregularity of the curvature. Under the Ricci flow, regions in the manifold of positive sectional curvature tend to shrink and regions of negative sectional curvature tend to expand and spread out.
%\yuyao{singularity undefined.}
Singularities usually occur while deforming a Riemannian 3-manifold through the Ricci flow.  They appear in a small neighborhood of a surface in the 3-manifold. By removing the singularities (i.e., surfaces) and redefining the Ricci flow on the remaining pieces, one produces the Ricci flow with surgery on the manifold. Figure~\ref{fig:manifold} (b) and (c) illustrate the formation of a singularity and the `surgery' operation. The ground-breaking work of Perelman~\cite{Perelman2002-gs} shows that the Ricci flow with surgery captures the geometric decomposition of the 3-manifold. It solves the Geometrization Conjecture of Thurston and geometrically classifies all 3-manifolds.

Ricci flow enables a better understanding of the evolution and community structure of networks. In our heuristic thinking, a network is analogous to a discretization of high dimensional manifold (say a 3-manifold) and communities in the network are analogous to the components in the geometric decomposition of the 3-manifold. Since Perelman's work~\cite{Perelman2002-gs} proved that the Ricci flow is able to predict geometric components of a 3-manifold, it suggests that a discrete Ricci flow on the network should be able to detect the community structure. Just like in Hamilton-Perelman's work on Ricci flow, the cutoff number of iterations and threshold value for surgery in Ricci flow depend on individual networks.

% \begin{figure}[htbp]
%     \centering
%     \subfigure[Ricci flow on a manifold]{
%         \label{fig:manifold:subfig:manifold}
%         \includegraphics[width=0.3\columnwidth]{figures/rf_manifold}}
%     \subfigure[A graph of four communities]{
%         \label{fig:manifold:subfig:uncut}
%         \includegraphics[width=0.3\columnwidth]{figures/rf_graph_uncut}}
%     \subfigure[The graph after Ricci flow]{
%         \label{fig:manifold:subfig:cut}
%         \includegraphics[width=0.3\columnwidth]{figures/rf_graph_cut}}
%     \caption{An illustration of Ricci flow on a manifold and a graph.}
%     \label{fig:manifold} %% label for entire figure
% \end{figure}

\section{Theory and Algorithms on Discrete Ollivier Ricci Flow}
\label{sec:main-theory}
%\input{scientific_report-theory}

% \saiba{need check}
% Our contribution is to apply discrete Ricci flow for community detection on a graph. 
In this section, we introduce our discrete Ricci flow algorithm for community detection on the network. We started with the definition of Ricci curvature by Ollivier in Equation~\ref{eqn:ollivier-ricci}, for each node $x$ on a metric graph $G=(V,E,w)$, we define a mass distribution $m_x$ on $x$'s neighbor nodes. 
A discrete \emph{transport plan} is a map $A: V \times V \to [0,1]$ such that $A(u, v)$ is the amount of mass at vertex $v$ to be moved to vertex $u$. It satisfies $\sum_{v' \in V} A(u,v')=m_x(u)$ and $\sum_{ u' \in V} A(u',v)=m_y(v)$. The Wasserstein distance here $W(m_x,m_y)$ is defined as the minimum total weighted travel distance to move $m_x$ to $m_y$, i.e., 
$ W(m_x,m_y)= \inf \{\sum_{u,v \in V} A(u,v) d(u,v) \} $.
%Let the minimal cost of transportation to move the mass from $m_x$ to $m_y$ to be $W(m_x,m_y)$ (i.e., the Wasserstein distance or earth mover distance), 
The discrete Ricci curvature on a network edge $xy \in E$ is defined as 
\begin{equation*}
    \kappa_{xy} = 1 -\frac{ W(m_x,m_y)}{d(x,y)},
\end{equation*}
where $d(x,y)$ is the length of the shortest path between nodes $x$ and $y$. 

Under this definition, if two nodes $x$ and $y$ are from different communities, their neighbor nodes tend to have fewer common neighbors, hence the best way to move $m_x$ from $x$'s neighbors to $m_y$ in $y$'s neighbors is to travel along the edge $xy$. Because of this, the Wasserstein distance is necessarily larger than the length of $xy$, which leads to negative Ricci curvature.
Alternatively, nodes within the same community tend to share neighbors or have shortcut between neighbors, thus have a Wasserstein distance no greater than $d(x, y)$. Therefore intra-community edges are mostly positively curved. See Figure~\ref{fig:curvature} for examples of network edges of positive, zero and negative curvatures.

Note that the probability distribution $m_x$ for $x \in V$ needs to be specified. In previous work~\cite{Lin2011-wk}, the probability distribution is uniform on $x$'s neighbors.
In this paper, we suggest a more general family of probability distributions $m^{\alpha,p}_x$, with two parameters: $\alpha \in [0,1]$ and power $p \ge 0$: 
\begin{equation*}
m^{\alpha,p}_x(x_i)=
\begin{cases}
\alpha & \mbox{\ if } x_i = x\\
\frac{1-\alpha}{C}\cdot  \exp(-d(x,x_i)^p) & \mbox{\ if } x_i \in \pi(x)\\
0 & \mbox{\ otherwise }.
\end{cases}
\end{equation*}
%Let $\pi(x)$ be $x$'s neighbors, $\alpha \in [0,1]$ \saiba{explain what's $\alpha$ for?} and a power parameter $p$, the probability measure $m^{\alpha,p}_x$ associated to the node $i \in V$ is inversely proportional to the edge weight connected to $x$:
% $$
% m^{\alpha,b,p}_x(x_i) = \left\{
%         \begin{array}{lll}
%             e^{-\alpha \eta(x, i)}/C & \quad
%           \mbox{if } i \sim x \\ 0  & \quad  \mbox{otherwise}
%         \end{array}
%     \right.
% $$
Here $C=\sum_{x_i \in \pi(x)} \exp(-d(x,x_i)^p)$ is a normalization factor and $\pi(x)$ is the set of neighbors of $x$. The parameter $\alpha$ determines the probability to remain at $x$. The power parameter $p$ determines how much we want to discount the neighbor $x_i$ of $x$ with respect to the weight $d(x,x_i)$. When $p=0$, the probability measure is uniform on all neighbors of $x$ as suggest in~\cite{Lin2011-wk}. For a large $p$, the neighbors that are far away from $x$ are aggressively discounted.

% if $x, y$ are within the same community, they either share a lot of neighbors or their neighbors are connected by other shortcut edges, hence the optimal transport distance is smaller. The Ollivier's Ricci curvature is one minus the ratio of the optimal transport distance and the edge length of $xy$. Intra-community edges are positively curved while inter-community edges are negatively curved.  

The discrete Ricci flow algorithm on a network is an evolving process. In each iteration, we update all edge weights simultaneously by the following flow process:

$$ w_{xy}^{(i+1)} = d^{(i)}(x,y) - \kappa_{xy}^{(i)} \cdot d^{(i)}(x,y), $$
where $w_{xy}^{(i)}$ is the weight of the edge $xy$ at the $i$-th iteration, and $\kappa_{xy}^{(i)}$ is the Ricci curvature at the edge $xy$ at the $i$-th iteration, and $d^{(i)}(x,y)$ is the shortest path distance on the graph induced by the weights $w_{xy}^{(i)}$. Initially $w_{xy}^{(0)}=w_{xy}$ and $d_{xy}^{(0)}=d_{xy}$. The detailed algorithm is presented in supplementary information.
%Namely, $d_{ij}(m)$ is $\min\{\sum_{k=0}^n w_{x_kx_{k+1}}(m): \text{the edge path \{$x_0x_1, x_1x_2, ..., x_{h-1}x_h\}$ joins vertices $x_0=i$ and $x_h=j$}\}.$ \jie{notation}
%To begin, for each edge $x, y$, we introduce a mass distribution at the neighborhood of node $x$ and another distribution at the neighborhood of $y$. 
% Different from Ollivier and Lin-Yau's definition, we take the distribution used in the optimal transport computation to be exponentially dependent on edge weights. That is, the probability at a neighbor $u$ of $x$ is proportional to $\exp(-w(x, u)^{\alpha})$, where $\alpha$ is a power parameter. 
% Thus a faraway neighbor of $x$ has a small amount of mass while a nearby neighbor carries more mass. Compute the Wasserstein distance at each edge and use it to replace the current edge weight.
%The process stops when the edge weights stabilize. 
%The initial edge weight for a graph is 1 for each edge. % to produce a new edge weight w
%Then we run Ricci flow on the network, in which an edge's length is reduced by a term weighted by its curvature.

This discrete Ricci flow process expands negatively curved edges and shrinks positively curved edges. Eventually, nodes connected by intra-community edges are condensed and inter-community edges are stretched. By this effect, a simple thresholding procedure can easily separate different communities. This is termed network `surgery' when edges of large weights (likely inter-community edges) are removed after several Ricci flow iterations (usually 10 to 15 iterations). See Figure~\ref{fig:karate_demo} as an example for the surgery process. For networks with hierarchical community structures, we may perform multiple rounds of network surgery and Ricci flow to fully separate communities at different scales.

\section{Results}

\subsection{Theoretical Results}
We can prove rigorously that the Ollivier-Ricci flow with respect to the specific choice of $\alpha=0$ and $p=0$ can successfully detect community structure for the following $G(a,b)$ family of graphs (Please refer to supplementary information for further detail). Take the complete graph on $b+1$ vertices $p_1, ..., p_{b+1}$ and $b+1$ complete graphs $C_1, ..., C_{b+1}$ on $a+1$ vertices. Take a vertex $u_i$ from each $C_i$ and identify $u_i$ with $p_i$. The resulting graph is $G(a,b)$. %from a complete graph on $b+1$ vertices by replacing each vertex by a complete graph on $(a+1)$ vertices.
For $a>b$, this is a highly symmetric graph with a clear community structure -- each copy of $C_i$ is a community and there are $b+1$ of them. Between any two communities $C_i, C_j$, there is only one edge $u_iu_j$ joining them. This community structure can be detected by the Ollivier-Ricci flow with respect to the Ollivier-Ricci curvature $K_0$ corresponding to $\alpha=0, p=0$ in Section~\ref{sec:main-theory}. More precisely, the Ollivier-Ricci curvature $K_0$ is associated with the
probability distribution $\mu_x$ % at a vertex $x$ depending only on the vertex degree $d_x$ of $x$, i.e., $\alpha=p=0$ case in \S3. More precisely, 
such that $\mu_x(y)=1/d_x$ if $y$ is adjacent to $x$ and $\mu_x(y)=0$ otherwise. In this case, we are able to compute explicitly the Ollivier-Ricci curvature at the $n$-th iteration of the Ricci flow and confirm how the weights of the network edges evolve over time. 

\begin{theorem} The Ricci flow associated to the Olivier $K_0$-Ricci curvature detects the community structure on $G(a,b)$ if $a > b \geq 2$, namely, the weight of the intra-community edges shrink asymptotically faster than the weight of the inter-community edges.
\end{theorem}\label{thm:convergence}
\begin{proof}
Please refer to supplementary information.
\end{proof}

\subsection{Experimental Results}

In this section, we explain the model networks and real-world datasets used to evaluate the community detection accuracy of our method. For the model network, we tested the growing geometrical network model with emergent complex geometry(GNet), and two models that provides community labels: the standard and widely used stochastic block model(SBM), and the Lancichinetti-Fortunato-Radicch benchmark model(LFR) that generates graphs of power-law degree distributions. For real-world datasets, we picked $6$ different community graphs that come with \emph{ground-truth community} labels. 
More detailed experiments can be found in supplementary information.

\subsubsection{Model Networks and Real World Datasets}
% \saiba{move to supplementary?} 

\paragraph{Stochastic Block Model}
%\label{ssub:sb_model}
The \emph{stochastic block model} (SBM) is a probabilistic graph model~\cite{JMLR:v18:16-480}. A graph following the stochastic block model has $n$ vertices, which are partitioned into $k$ communities. Two nodes within a community are connected with probability $p_{intra}$ while two nodes in different communities are connected with probability $p_{inter}$, $p_{intra}>p_{inter}$.
%The nodes are partitioned into $k$ disjoint communities in terms of an $n\times 1$ vector $\vec{c}$ where $\vec{c(l)}$ gives the community label of node $l$. 

%The edges are randomly selected with probabilities from a $k\times k$ symmetric matrix $M$. A node in community $i$ and a node in community $j$ have an edge with probability $M_{ij}$, independently from other edges.

%The connection between nodes is decided independently at random by a $k\times k$ symmetric matrix $M$, where $M_{ij}$ gives the probability that a node in community $i$ is connected to a node in community $j$ with probability $M_{ij}$. 
% In our evaluation we construct the SBM by module provided by NetworkX\footnote{https://networkx.github.io/documentation/latest/reference/generated/networkx.generators.community.random_partition_graph.html#networkx.generators.community.random_partition_graph}.

% Let $n$ be the number of vertices of the network with $k$ communities, $p=(p_{1},\cdots ,p_{k})$ be a probability vector, and $W$ be a $k\times k$ symmetric matrix where $w_{ij}\in [0,1]$. An $n-$node undirected graph $G=(V,E)$ composed of $k$ communities is drawn by an $n-$dimensional random vector $X$ with components $i.i.d.$ distributed under $p$. Each pair of nodes $u, v$ are connected with probability $w_{X_{u},X_{v}}$.
% \saiba{need rewrite}
% \jie{James and Yu-Yao, this definition is very confusing. First, $w_{ij}$ is never mentioned later. Second, what does it mean that $G$ is drawn by a random vector? }

% \saiba{TODO: add # edge in fig}

%\subsection{LMF Model}

\paragraph{Lancichinetti-Fortunato-Radicch Model} % (fold)
%\label{ssub:lfr_benchmark}
The Lancichinetti-Fortunato-Radicch (LFR) benchmark~\cite{lancichinetti2008benchmark} generates undirected unweighted networks with non-overlapping communities. The model produces networks with both degree and community size satisfying power-law distributions. This model is also commonly used to evaluate community detection algorithms~\cite{yang2016comparative}. 

%The generative algorithm of the LFR graph is summarized as the following steps:
% \begin{enumerate}
%    \item Given the number of vertices $n$, maximum degree  $k_{\max}$, average degree $k$, and parameter $\gamma$ as the exponent in the power-law degree distribution, we can generate a \emph{degree sequence} $d_1, d_2, \cdots, d_n$. Construct a network $G = (V,E)$ with $|V|=n$ and degree sequence $d_1, d_2, \cdots, d_n$ by using the configuration model~\cite{molloy1995critical}. In particular, we generate for node $i$ $d_i$ half edges, and then randomly matches the `half-edges'. 
%    The graph has $n$ vertices and the $i$-th vertex has degree $d_i$.
%    
%    \item Take a sequence of community sizes from a power-law distribution with exponent $\beta$, such that the sum of the sequence equals $|V|$. The sequence is bounded by a defined maximal community size $s_{\max}$. Nodes are then assigned to these communities.  
%    
%    \item Set a mixing parameter $\mu$ as the ratio between external degree over the total degree of each node, i.e. each node connects a fraction $1-\mu$ of its edges with other nodes inside its community and a fraction $\mu$ of edges with nodes outside its own community.
%
%    \item Rewire edges so that the fraction of external neighbors is well approximated by $\mu$. That is, take a random pair of edges and switch their endpoints. Notice that the degree of each node is preserved.
% \end{enumerate}

% subsubsection lfr_benchmark (end)

\vspace*{3mm}
\noindent\textbf{Emergent Geometrical Network Model}
The emergent geometrical network model~\cite{Bianconi2014-cw,wu2015emergent} (GNet) describes a growing network with high clustering coefficient using the triadic closure property. It is observed to have non-trivial community structures. One version described in~\cite{wu2015emergent} could grow a geometric network. It is composed of the skeleton of a simplicial complex in which a set of $2$-simplices are glued together properly. The generation of this model is controlled by the designated number $m$ of 2-simplices glued along a 1-simplex (edge), and the probability $p$ of connecting two nodes with hop distance $2$.

\paragraph{Real World Datasets}
%\label{sub:real_data}
% For real world datasets, we choose networks that provide ground truth communities from KONECT\footnote{https://konect.uni-koblenz.de}, UCI network data repository\footnote{https://networkdata.ics.uci.edu} and Stanford Network Analysis Project\footnote{https://snap.stanford.edu/data/index.html}. The statistics of the real world datasets are summarized in Table~\ref{tbl:realworld_dataset}. In the followings, we briefly describe the datasets.

For real world datasets, we choose networks that provide ground truth communities from KONECT\cite{Kunegis2013-bg}, UCI network data repository and Stanford Network Analysis Project~\cite{snapnets}. The statistics of the real world datasets are summarized in Table~\ref{tbl:realworld_dataset}. In the followings, we briefly describe the datasets.

% \saiba{ sec4.2 in http://dm.uestc.edu.cn/wp-content/uploads/paper/Community-Detection-based-on-Distance-Dynamics.pdf}
\begin{itemize}
\item \emph{Karate club network}
The Karate club network dataset was collected from the members of a university karate club by Wayne Zachary in 1970s%~\cite{zachary1977information}
. The network is undirected in which nodes represent members of the club, and edges represent ties between two members. This dataset is generally used to find the two groups of people into which the karate club fission after a conflict between two faculties. 
% Zachary correctly assigned all but one members of the club to the groups they belonged. According to the original paper by Zachary, the incorrectly assigned member is regarded as having an ambiguous inclination towards any of the two groups. In our experimental result, we correctly predict all member's group except the ambiguous one and another.
% \yuyao{please check}\saiba{which one is that?}

\item \emph{American college football network}
The American college football network is a representation of the schedule of Division I games during the season Fall $2000$ and was previously used for community detection by Girvan and Newman~%\cite{Girvan2002-cq}
. Each node represents a football team and each edge indicates a game between two teams. The community structure of the network is given by partitioning the teams into $12$ conferences. Games held between teams of the same conference are held more frequently than games played between different conferences. 

\item \emph{Political books network}
This is a network of books about US politics published around the time of the $2004$ presidential election and sold by the online bookseller \emph{Amazon.com}. Edges between books represent frequent co-purchasing of books by the same buyers. 

\item \emph{Political blogs network}
The $2004$ U.S. Presidential Election was notably influenced by blogs. The political blogs network dataset was collected by Adamic and Glance~%\cite{adamic2005political}
 in $2005$. The posts published by either liberal or conservative bloggers are represented by nodes. Any two nodes are connected by an edge if one of them is cited by the other. 

\item \emph{Ego-network from Facebook}
The ego-network dataset
%~\cite{leskovec2012learning}
 consists of `friend circles' of one anonymous user and his/her friends on Facebook. The network forms friend circles such as family members, high school friends or other friends that are `hand labeled' by the user. To normalize the influence of users belongs to multiple circles, we treat the overlaid circle as a new circle.

\item \emph{Email-EU-core network}
The Email-Eu-core network was formed by the email contacts between members of a large European research institution. The members are represented by nodes where any pair of nodes are connected by an edge if they have had contacts through e-mail. Each individual belongs to exactly one of 42 departments at the research institute.%~\cite{leskovec2007graph}.

\end{itemize}

\begin{table}
\centering
\caption{Real World Dataset}
\label{tbl:realworld_dataset}
\begin{tabular}{lrrrrrr}
\hline
Datasets    & V     & E      & \#Class & AvgDeg  & CC     & Diameter \\
\hline
\hline
Karate Club & 34    & 78     & 2       & 4.5882  & 0.5706 & 5        \\
Football    & 115   & 613    & 12      & 10.6608 & 0.4032 & 4        \\
Polbooks    & 105   & 441    & 3       & 8.4000  & 0.4875 & 7        \\
Polblogs    & 1222  & 16714  & 2       & 27.3552 & 0.3203 & 8        \\
FB-Ego      & 792   & 14025  & 24      & 35.417  & 0.483  & 10       \\
Email-EU-core& 1005 & 16064  & 42      & 31.968  & 0.450  & 7        \\                
% Cora        & 2485  & 5069   & 7       & 4.0797  & 0.2376 & 19       \\
% Citeseer    & 2120  & 3679   & 6       & 3.4708  & 0.1697 & 28       \\
%Blogcatalog & 10312 & 333983 & 39(1059 Overlap) & 64.776  & 0.4760  & 5 \\
\hline   
\end{tabular}
\end{table}

%We have applied the Ricci flow method to many networks from the real world and artificial networks generated by two different methods: the stochastic block model (SBM)~\cite{abbe2017community}, in which two nodes within a community are connected with probability $p_{intra}$ while two nodes in different communities are connected with probability $p_{inter}$, $p_{intra}>p_{inter}$, and the Lancichinetti-Fortunato-Radicchi (LFR) benchmark~\cite{lancichinetti2008benchmark} (a community model with power-law degree distribution). 

\subsubsection{Experimental Results}

To evaluate the clustering accuracy of our algorithm, we tested the clustering result with two different metrics: Adjusted Rand Index (ARI) and modularity. ARI measures the accuracy of clustering result with the \emph{ground truth} clustering. Modularity quantifies the strength of the community structure of a given graph without the need of ground-truth clustering. 

\paragraph{Clustering Accuracies} 

The Clustering accuracies of applying discrete Ricci flow for $50$ iterations is shown in Figure~\ref{fig:evaluation}. In Figure \ref{fig:evaluation}(a) and Figure \ref{fig:evaluation}(b), the parameters $p_{inter}/p_{intra}$ of the SBM and $\mu$ of the LFR indicate the magnitude of community structure of the models respectively. In both models, higher parameter values in $x$-axis indicate weaker community structures.
% In both models there is a parameter that decides whether the community structure is significant ($x$-axis). 
We choose the adjusted Rand index (ARI)~\cite{hubert1985comparing} as the quality measure for the clustering accuracy compared with the ground truth, as shown in the vertical axes. The ARI score represents the agreement of partitioned node pairs in ground truth communities and clustered communities. The higher the ARI score is, the more accurate our detected communities are. The results of Ricci flow algorithm show robust detection of community structures that compares favorably with prior algorithms -- with a sharp phase transition from nearly $100\%$ accuracy for SBM models with $p_{inter}/p_{intra}=0.5$ (almost all nodes separated correctly) to nearly $0\%$ accuracy for models with $p_{inter}/p_{intra}=0.55$ (meaning the non-existence of community structure). Similar results have been observed with modularity.

%To demonstrate the effectiveness of the discrete Ricci flow algorithm for community detection, we test the algorithm on both model networks and real-world networks. For model networks, we test Stochastic Block Model (SBM), in which two nodes within a community are connected with probability $p$ while two nodes in different communities are connected with probability $q$ with $p>q$, and Lancichinetti-Fortunato-Radicchi (LFR) benchmark~\cite{lancichinetti2008benchmark} (a community model with power-law degree distribution).  We choose Adjusted Rand Index (ARI)~\cite{hubert1985comparing} as a quality measure for the clustering accuracy. The result is present in Figure~\ref{fig:evaluation}. Discrete Ricci flow is shown to provide nearly perfect clustering result when community structures exist. Also, extensive comparison tests on real networks show that our algorithm is competitive with previously proposed ones.

\begin{figure}[htbp]
    \centering
    % \subfigure[SBM]{
    %     \label{fig:evaluation:subfig:SBM}
    %     \includegraphics[width=0.45\columnwidth]{figures/ARI-SBM_dense-all}}
    % \subfigure[LFR]{
    %     \label{fig:evaluation:subfig:LFR}
    %     \includegraphics[width=0.45\columnwidth]{figures/ARI-LFR-all}}
    % \subfigure[Real Networks]{
    %     \label{fig:evaluation:subfig:real}
    %     \includegraphics[width=0.44\columnwidth]{figures/realdata-ARI}}
    \includegraphics[width=1\columnwidth]{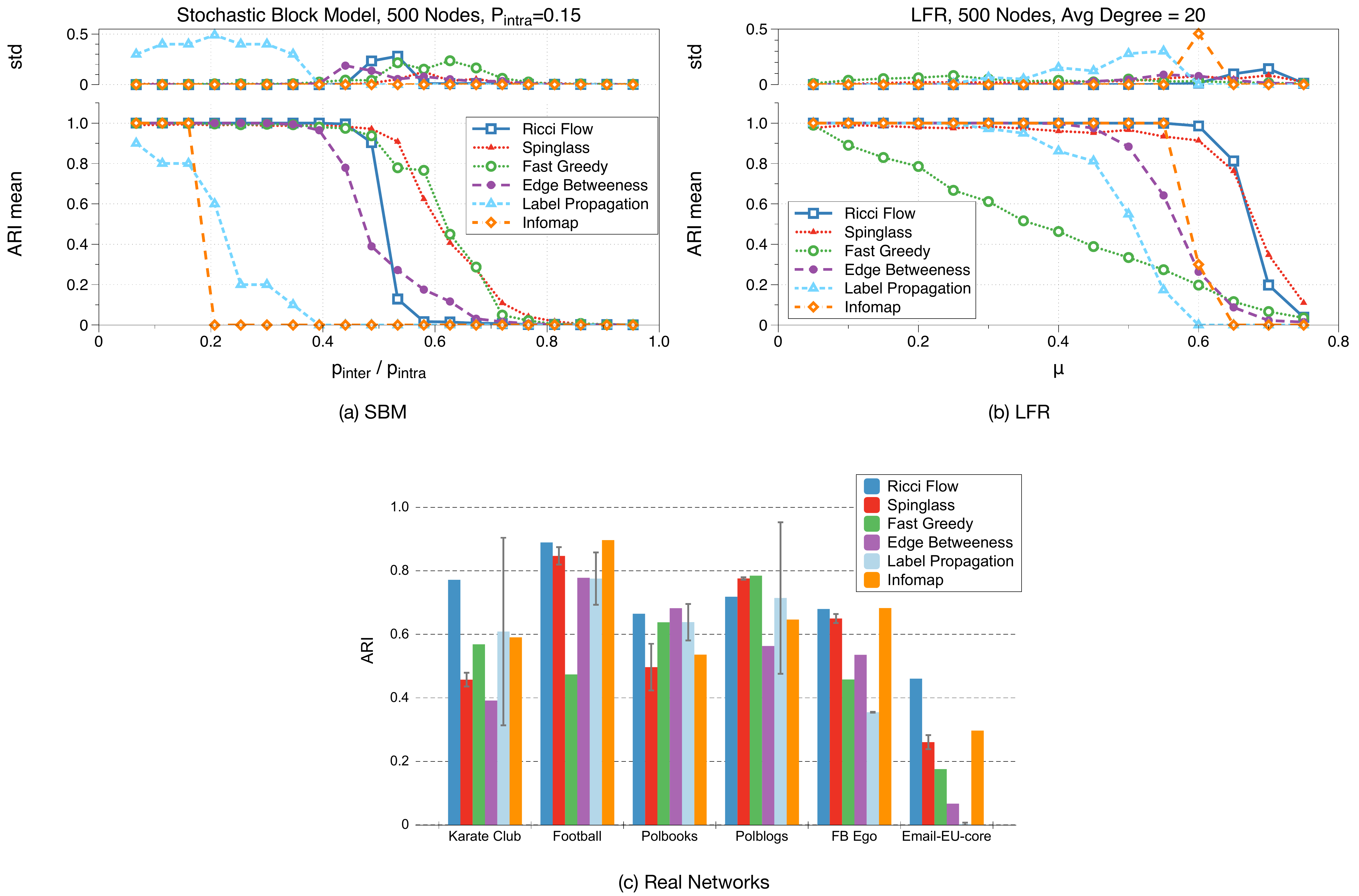}
    \caption{The accuracy of the Ricci flow method for community detection on model networks. The accuracy is measured by the adjust Rand index (ARI) and each data point is the average of $10$ model graphs. In (a), we tested on the stochastic block model (SBM) with $500$ nodes and two communities of the same size. A graph with low $p_{inter}/p_{intra}$ ratio has more distinctive communities. Our method is shown to have perfect accuracy with $p_{inter}/p_{intra} < 0.5$. In (b), for Lancichinetti-Fortunato-Radicchi (LFR) Model, we set the graph to have $500$ nodes, average degree of $20$, with $38$ communities. LFR can produce graphs with power-law degree distribution with communities of different sizes. The magnitude of community structure is controlled by $\mu$, the ratio of inter-community edges with intra-community edges. Again, our method produces the best accuracy among all methods. In (c), for non-deterministic algorithm Spinglass and Label Propagation, the accuracy are averaged over $10$ runs. }
    \label{fig:evaluation}

\end{figure}

% On the same LFR graph as in Figure~\ref{fig:iteration} with $1000$ nodes, $9539$ edges, $30$ communities and $\mu=0.4$, the Ricci flow totally performed $20$ iterations (as the blue line in Figure~\ref{fig:iteration}). 

To remove the singularities generated during the Ricci flow, we applied the surgery which removed edges with weight greater than an intermediate cutoff threshold for every $5$ iterations during the whole $50$ iteration process. The clustering accuracy results under different accuracy metrics are shown in Figure~\ref{fig:measures}. In Figure~\ref{fig:measures}(a), when the (final) cutoff threshold is set between $1$ and $0.47$, we have a perfect clustering result of detecting all $30$ communities, and this is correctly captured by ARI with the highest possible score $1.0$. 
(In classical case of Hamilton-Perelman Ricci flow on 3-manifolds, the time to do surgery depends on individual manifolds)
% We remark that there is no cutoff threshold we can set for removing edges. This is very similar to the classical case of Hamilton-Perelman Ricci flow on 3-manifolds where the time to do surgery depends on individual manifolds. 
For modularity, the trend of capturing the perfect clustering accuracy result is similar to ARI (before the cutoff threshold $0.47$), but its highest score occurs with a cutoff threshold of $0.275$, which detected $290$ communities. With this connection that ARI and modularity tend to capture the communities in the same trend, hence for networks without community labels such as GNet, a cutoff threshold is suggest to be when modularity first hits the plateau of the curve, for example with cutoff at $3.2$ in Figure~\ref{fig:measures}(b). 
This cutoff threshold also gives us a hint to detect hierarchical community structures. In Figure~\ref{fig:GNet-planar}, layered community structures are revealed by applying different cutoff thresholds after $20$ iterations of discrete Ricci flow processes. 

% We also observe that ARI is a more reasonable measure than NMI and modularity as it discourages overly aggressive removal of edges -- when most edges are removed, the number of detected communities is almost the same as the number of nodes in graph (every node form a community in this case), only ARI can correctly capture this behavior and yields the lowest score of $0$. In the following section, the clustering accuracy result is presented in ARI with the best possible cutoff threshold.

\begin{figure}[htbp]
    \centering
    % \subfigure[LFR]{
    %     \label{fig:measure:subfig:LFR}
    %     \includegraphics[width=0.51\columnwidth]{figures/Measures.pdf}}
    % \subfigure[GNet]{
    %     \label{fig:measure:subfig:GNet}
    %     \includegraphics[width=0.43\columnwidth]{figures/m2p09.pdf}}
    \includegraphics[width=1\columnwidth]{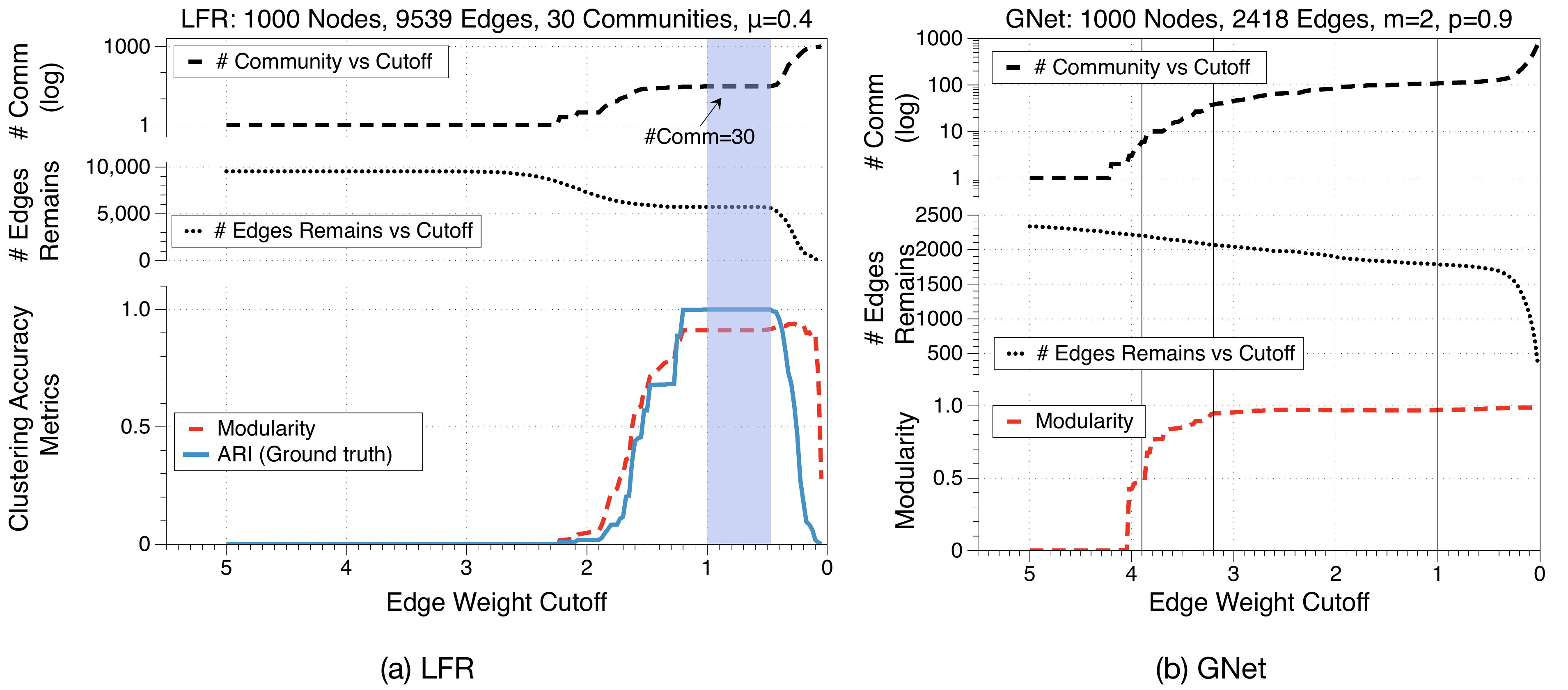}
    \caption{A comparison of clustering accuracy on an LFR graph after $50$ iterations and GNet after $20$ iterations of the Ricci flow with different final edge weight cutoff thresholds. In (a), with cutoff threshold set between $1$ and $0.47$ as the range highlighted in blue, we detected all communities correctly. In (b), we chose the cutoff threshold to be the turning point of modularity at $w=3.2$ as the middle vertical line. To compare the communities detected with different cutoff thresholds, two extra cutoff $w=3.9$ and $w=1$ are added. The detected communities are shown in Figure~\ref{fig:GNet-planar}.}
    \label{fig:measures} %% label for entire figure
\end{figure}

\begin{figure}[htbp]
    \centering
    \includegraphics[width=1\columnwidth]{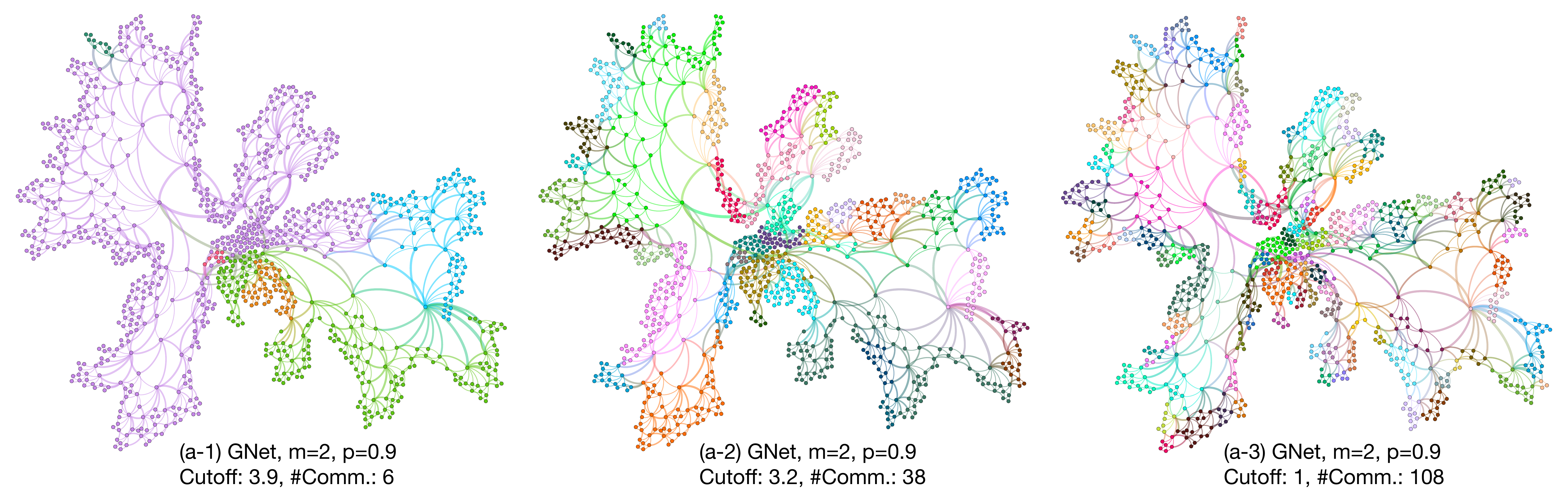}
    \caption{Illustrations of communities detected by discrete Ricci flow with three different cutoff weights of a GNet planar model with $1000$ nodes, $m=2$, and $p=0.9$. With different cutoff thresholds (labeled as vertical lines in Figure~\ref{fig:measures}(b)), we are able to detect communities in a hierarchical manner.}
    \label{fig:GNet-planar}
\end{figure}

\subsubsection{Comparison with Other Methods}

We compared our result with the community detection algorithms such as Modularity based Fast Greedy algorithm~\cite{Clauset2004-bp}, Label Propagation~\cite{Raghavan2007-wf}, Infomap~\cite{Rosvall2008-ns}, Spinglass~\cite{Reichardt2006-zs}, and Edge Betweenness~\cite{Girvan2002-cq} (by iGraph: \url{http://igraph.org/python/}) with Adjusted Rand Index (ARI) as the accuracy metric. 

We first tested community detection algorithms on a simple graph model SBM with $500$ nodes, $6800$ edges and two even size communities in Figure~\ref{fig:evaluation}(a). We fixed $P_{inter}=0.15$ and tested the mixing ratio $P_{intra}/P_{inter}$ from $0.1$ to $0.9$. For SBM, beside label propagation method and Infomap, most of the algorithms perform well when the mixing ratio is below $0.5$.

For LFR graphs, Ricci flow and Spinglass outperform all other methods in our experiments (Figure~\ref{fig:evaluation}(b)). Compared to the accuracy of $95\%$ for Spinglass, Ricci flow is more stable with nearly perfect accuracy for most of the values of $\mu$.% since it is a deterministic algorithm. 

We also evaluated community detection algorithms on different real-world datasets. In Figure~\ref{fig:evaluation}(c), Ricci flow shows competitive or better results in Karate club, Football, Polbooks, and Polblogs datasets.
% In Cora and Citeseer dataset, since these two graphs are with low average degree (around 4) and low clustering coefficient, the community structure are weaker, hence none of the algorithm can provide better accuracy more than $0.3$.

One key factor of a community structure is the density of connections within communities, the community structure is stronger if nodes in one community are more densely connected. In Figure~\ref{fig:degree}, we tested Ricci flow and spinglass on LFR graphs with different average degree settings. 
The results show that with a higher average degree (higher edge density within communities) both algorithms provide better clustering results.

\begin{figure}[htbp]
    \centering
    % \subfigure[Ricci Flow]{
    %     \label{fig:degree:subfig:RF}
    %     \includegraphics[width=0.48\columnwidth]{figures/ARI-RF-Degree}}
    % \subfigure[Spinglass]{
    %     \label{fig:degree:subfig:SG}
    %     \includegraphics[width=0.48\columnwidth]{figures/ARI-SG-Degree}}
    \includegraphics[width=1\columnwidth]{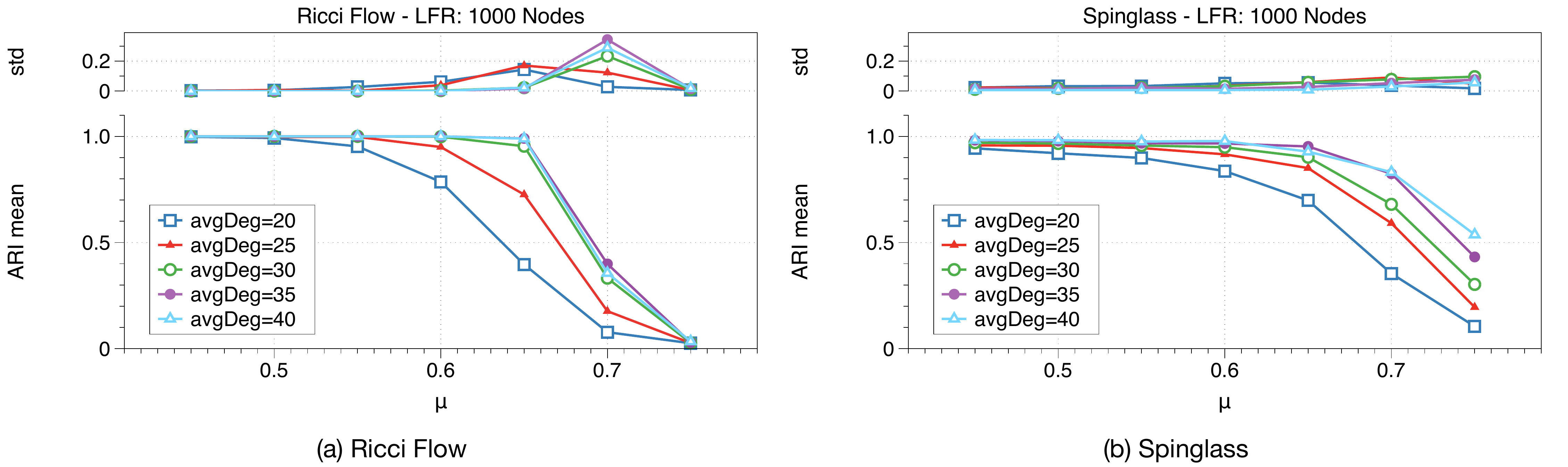}
    \caption{ Ricci flow and spinglass algorithms on LFR graphs with different average node degrees. With a higher average degree which implies higher edge density within communities, both algorithms provide better clustering results.}
    \label{fig:degree} %% label for entire figure
\end{figure}

\section{Conclusion}

In this paper, we have introduced geometric tools to investigate the community structures on complex networks.  The basic idea is to consider networks as geometric objects and use the notion of curvature and curvature guided flow to decompose networks. In classical mathematics, Ricci curvature and Ricci flow are among the most important tools for analyzing and decomposing manifolds according to their geometric and topological properties. What is interesting is that the corresponding discrete counterparts are shown to be powerful for detecting community structures. Interesting future works include improving the theoretical understanding of discrete curvature on graphs and applying our methods for real-world applications. 

\subsection*{Acknowledgement.} The authors would like to acknowledge support by NSF DMS 1737876, NSF DMS 1405106, NSF DMS 1811878, NSF FRG 1760527, NSF DMS-1737812,  NSF CNS-1618391 and  NSF CCF-1535900. We thank Xianfeng David Gu and Wujun Zhang for discussions.

\subsection*{Author Contributions.} All authors have contributed to the design of algorithms and writing of the article. Yu-Yao Lin and Chien-Chun Ni carried out implementation of the algorithm as well as evaluations under different models.  

\subsection*{Competing interests} 
The authors declare no competing interests.

\subsection*{Data availability statement}
The datasets generated during and/or analyzed during the current study are available from the corresponding author on reasonable request. Code for graph Ricci curvature and Ricci flow computation is available on github (\url{https://github.com/saibalmars/GraphRicciCurvature}).

\subsection*{Correspondence} Please send correspondences to Jie Gao, Department of Computer Science, Stony Brook University, Stony Brook, NY 11794. Email: jgao@cs.stonybrook.edu.

\newpage
\appendix
\section*{Supplementary Materials}

\renewcommand{\figurename}{Supplementary Figure S}

\section{Theoretical foundation}

We model a network as an unweighted graph $G=(V, E)$ with node set $V$ and edge set $E$.
The community detection problem on $G$ is to find a collection of edges, called \emph{inter-communities edges}, such that removing these edges from $G$ produces a collection of connected subgraphs $C_1, ..., C_n$, called \emph{communities}, with the following properties (1) nodes in each community $C_i$ are heavily connected by \emph{intra-community edges} among themselves and (2) nodes in different communities are sparsely connected. It has been observed in many real-world networks (e.g., social, biological, the Internet) that communities exist. Detecting communities in a network has become an important and fundamental problem in complex networks with many applications. However, from a mathematical point of view, a precise and universally accepted definition of community structure is still lacking.

Our approach to define communities in a network uses the notion of Ricci curvature in geometry and the mathematical theory of optimal transport. Consider a graph $G$ as a transportation network, edges between different communities are heavily traveled compared to edges within a community, since traffic prefers to take shortest paths whenever possible. Specifically, one can quantify the amount of traffic through an edge using the transportation cost from the optimal transportation theory.

We will use the following definitions and conventions. Two vertices (or nodes) $i, j \in V$ are adjacent, denoted by $i \sim j$, if they are the endpoints of an edge.  In this case, we denote the edge by $ij$.
 The edge weight $w: E \to \R_{\geq 0}$ on $G$ assigns each edge $ij$  a non-negative number $w_{ij}$. The triple $(V, E, w)$ is called a  weighted graph or a metric graph.
A path from node $a$ to node $b$ is a collection of edges $e_i=v_iv_{i+1}$ for $i=0,1,2, ..., n-1$ such that $v_0=a$ and $v_n=b$.  The \it length \rm of the path $\{ e_0, e_1, ..., e_{n-1}\}$ is defined to be $\sum_{i=0}^{n-1} w_{i(i+1)}$. The path is said to have $n$ hops. 

\subsection{The optimal transport problem}
The original optimal transport problem considered by G. Monge in 1781 is to minimize the transportation cost to move iron ores in different mines to a collection of factories which consume the iron ores. Mathematically, the problem is formulated as follows.  To begin, let us briefly recall the notion of metric spaces and Borel measures on a metric space.  A metric space is a pair $(X, d)$ where $X$ is a set and  $d$ is a distance function $d: X \times X \to \mathbf R_{\geq 0}$ with the following properties:
\begin{itemize}
    \item $d(x,y)=0$ if and only if $x=y$;
    \item $d(x,y)=d(y,x)$;
    \item $d(x,y)+d(y,z) \geq d(x,z)$ for all $x,y,z \in X$.
\end{itemize} 

Given a metric space $(X,d)$, a Borel set in $X$ is obtained from open sets in $X$ through the operations of countable union, countable intersection, and relative complement.   
A Borel probability measure $\mu$ on a metric space $(X, d)$ assigns each Borel set $A$ a non-negative real number $\mu(A)$ such that 
\begin{itemize}
    \item $\mu(X)=1$;
    \item if $A_, ...., A_n, ...$ are disjoint Borel sets then $\mu(\cup_{i=1}^{\infty} A_i)=\sum_{i=1}^{\infty} \mu(A_i)$.
\end{itemize}  

In our case of a finite graph $G=(V,E)$, we take $X$ to be the vertex set $V$. Every subset of $X$ is Borel. A Borel probability measure can be identified with a function $\mu: V \to [0,1]$ such that $\sum_{i \in V} \mu(i)=1$.  Given an edge weighted graph $(V,E, w)$ with $w_{ij}>0$ for all edges, one introduces a metric $d$ on vertex set $V$ by
\begin{equation}\label{dis}
d(v, v') =\min\{  \sum_{i=0}^n   w_{ k_i, k_{i+1}}:  k_i \in V,  k_0=v, k_{n+1}=v', k_i \sim k_{i+1}\}
\end{equation}  where the minimum is taken over all edge paths from $v$ to $v'$. We call $d$ the induced metric from the edge weight $w$.  Note that by definition, $d(v, v')+d(v', v) \geq d(v, v'')$ for any three vertices $v, v', v'' \in V$.

Let $X$ and $Y$ be two metric spaces
and $\mu$ and $\nu$ be two probability Borel measures on $X$ and $Y$ respectively. Here $X$ and $Y$ stand for mines and factories and $\mu$ and $\nu$ are the amount of iron ores to be moved and consumed respectively. Let  $c : X \times Y  \to \R_{\geq 0}$
be a continuous function considered as the cost, i.e., the cost of transporting from location $x$ to location $y$ is $c(x,y)$. The function $c$  is usually taken to be the distance
$d(x, y)$  if $X=Y$ and the cost of transportation per-unit distance is constant.

A \it transport map \rm $T: (X, \mu) \to (Y, \nu)$ is a measure
preserving map, i.e., for any Borel set $A \subset Y$,  $\nu(A)=\mu(T^{-1}(A))$. Monge's formulation of the optimal transportation
problem is to find a transport map $T : X \to Y$ that realizes
the infimum 
$${\displaystyle \inf \left\{\left.\int
_{X}c(x,T(x))\,\mathrm {d} \mu (x)\;\right| \text{T is a
transportation} \right\}}.$$
A map $T$ that attains
this infimum (i.e. makes it a minimum instead of an infimum) is
called an ``optimal transport map". In this generality, there is
no mathematical theorem that guarantees the existence of the optimal transportation map. A breakthrough in optimal transportation problem was made by L. Kantorovich in 1942 \cite{kantorovich1942translocation}.  In this paper, he formulated the optimal transportation problem as a
linear optimization problem whose solution always exists.  He replaced the transportation map $T$ by the \it transportation plans \rm which are
probability measures $\gamma$ on the product space $X \times Y$  satisfying $\gamma(A \times Y)=\mu(A)$
and $\gamma(X \times B)=\nu(B)$ for all Borel sets $A$ and $B$.  The goal is to  find a transportation plan  $\gamma$
that attains the infimum cost $$W(\mu, \nu) ={\displaystyle \inf
\left\{\left.\int _{X\times Y}c(x,y)\,\mathrm {d} \gamma
(x,y)\right|\gamma \in \Gamma (\mu ,\nu )\right\}} ,$$ where
$\Gamma(\nu, \mu)$ denotes the collection of all transportation
plans. If $X=Y$, the quantity $W(\mu, \nu)$ is called the \it
Wasserstein distance \rm between two probability measures $\mu,
\nu$ on $X$. Kantorovich proved that the infimum is always
achieved by some transportation plan.

In our case of a finite weighted graph $G=(V, E, w)$ with two probability measures $\mu$ and $\nu$ on the vertex set $V$, one can reformulate Kantorovich's problem as follows.  Let $d$ be the associated distance function on $V$ defined by Equation (\ref{dis}).  A transportation plan $\gamma$ is given by a map $\gamma:  V \times V \to [0,1]$ such that $\sum_{i \in V} \gamma_{ij} =\mu_j$ and $\sum_{j \in V} \gamma_{ij} =\nu_i$ for all $i, j \in V$.  The goal is to find the minimum cost $$\min\{ \sum_{i, j \in V}  \gamma_{ij} d(i,j): \text{ $\gamma$ is a transportation plan}\}.$$  This is a linear programming problem and thus can be computed.
%and is usually solved by the Hungarian method.  
% See for instance \cite{kuhn}.

\subsection{Curvatures in classical differential geometry}

One of the central themes in modern geometry is the notion of curvature which measures how space
is curved. It was introduced by Gauss and Riemann over 190 years ago.  By an $n$-dimensional manifold, following Riemann, we mean a space which locally looks like the $n$-dimensional space.  A Riemannian metric on a manifold assigns each tangent space of the manifold a Euclidean metric. Manifolds together with  Riemannian metrics, i.e.,  Riemannian manifolds, are the main objects of study in modern geometry.  For instance, a smooth surface in the $3$-dimensional Euclidean space is a $2$-dimensional Riemannian manifold whose associated Riemannian metric is induced from the $3$-dimensional Euclidean space. These are the original objects investigated by Gauss.  For such a surface $S$,  the \it Gaussian map \rm from $S$ to the unit sphere sends a point on $S$ to the unit normal vector of $S$ at $p$.  The \it curvature \rm (or Gaussian curvature) of the surface at a point $p$ is the signed area distortion of the Gaussian map at $p$. To be more precise, it is the Jacobian determinant of the Gaussian map at $p$.
From this definition, one sees that the plane has zero curvature, the sphere has positive curvature, and the hyperboloid of one sheet has negative curvature. Gauss showed (Theorema egregium) that the curvature depends only on the induced Riemannian metric on the surface and does not depend on how the surface is embedded.
Riemann generalized Gaussian curvature to high dimensions as follows. For a Riemannian manifold $(M,g)$, the \it sectional curvature \rm assigns each 2-dimensional linear subspace $P$ in the tangent space  $T_pM$ of $M$ at $p$ a scalar (the Riemannian sectional curvature). The scalar is equal to the Gaussian curvature of the image of $P$  under the exponential map. A positively curved space tends to have small diameter and is geometrically crowded (e.g., a sphere). In contrast, a negatively curved space has an infinite fundamental group, a contractible universal cover, and is geometrically spreading out like a tree (in large scale).

% \begin{figure}[ht!]
% \begin{center}
% %\begin{tabular}{c}

% \includegraphics[width=0.75\textwidth]{figures/figure4.pdf}

% %\end{tabular}
% \end{center}
% \caption{Gaussian curvature and map, positive and negatively curved surfaces. \saiba{resolution}}
% %\label{555}
% \end{figure}

% \begin{figure}[htbp]
%     \centering
%         \label{fig:gaussian:subfig:map}
%         \includegraphics[width=0.4\columnwidth]{figures/GaussianMap}
%     \caption{Gaussian map sending a point to its unit normal vector.}
%     \label{fig:gaussian} %% label for entire figure
% \end{figure}

The \it Ricci curvature \rm assigns each unit tangent vector $v$ at a point $p$ a scalar which is the average
of the sectional curvatures of planes containing $v$. There are several characteristic properties
of the Ricci curvature which are used for defining the Ricci curvature on general metric spaces.
Namely, Ricci curvature controls how fast the volume of distance ball grows as a function
of the radius. It also controls the volume of the overlap of two balls in terms of their radii and the
distance between their centers.

There have been various approaches to generalize the notion of curvature to spaces which are not manifolds (e.g., a graph with edge weights). One of the important work of Ollivier \cite{Ollivier2009-yl}  Ricci curvature relates to optimal transport. Since optimal transport can be formulated on very general metric spaces with probability measure at each point, in particular on networks with edge weights and probability measures at each vertex. This is the approach we take for community detection using curvature and optimal transport.

%\begin{figure}[ht!]
%\begin{center}
%\begin{tabular}{c}
%
%\% \includegraphics[width=0.95\textwidth]{curvture1.pdf}
%
% \end{tabular}
%\end{center}
%\caption{ sectional curvature, Ricci curvature}
%\label{555}
%\end{figure}
\def\Ricci{\operatorname{Ricci}}

\subsection{Ollivier's work on Ricci curvature for general metric spaces with measures}

Ollivier's approach to Ricci curvature relies on the following key
theorem \cite{Ollivier2009-yl} which relates curvature with optimal transportation.  Let $(M^n,d)$ be an $n$-dimensional
Riemannian manifold with Riemannian distance $d$ whose Ricci curvature is $k$ and Riemannian volume measure is $\mu$. Fix $\epsilon>0$, let 
\begin{equation} \label{mass1} 
    m_x=\frac{\mu|_{B(x, \epsilon)}}{\mu(B(x, \epsilon))} 
\end{equation}
be the probability measure associated to $x \in M$. Then the
Wasserstein distance $W(m_x, m_y)$ is $$(1-k(x,y))d(x,y),\,
k_{xy}=\frac{\epsilon^2
k(v,v)}{2(n+2)}+O(\epsilon^3+\epsilon^2 d(x,y))$$ and
$v$ is the tangent vector at $x$ of the geodesic $xy$.
This shows that the classical Ricci curvature is related to the optimal transportation problem.
%For instance non-negative Ricci curvature implies volume of a ball is less than or equal
%to the volume of a ball of the same radius in the Euclidean space of the same dimension. Using the
%observation that Ricci curvature is related to the displacement convexity of entropy, Lott-Villani
%\cite{lott-villani}  generalizes Ricci curvature lower bounded spaces to metric spaces with a measure.
This also shows how to define
 Ricci curvature for general metric spaces with probability measures.

 \begin{definition} (Ollivier~\cite{Ollivier2009-yl})
 Given a metric space $(X, d)$ equipped with a
probability measure $m_x$ for each $x\in X$, the Ollivier's Ricci curvature
 along the shortest path $xy$ is
  \begin{equation} \label{ricci} k(x,y)
=1-\frac{W(m_x, m_y)}{d(x,y)} \end{equation}
where $W(m_x, m_y) $ is the Wasserstein distance with respect to the cost function $c(x, y)=d(x,y)$.  \end{definition}

\subsection{Ricci curvature for weighted graph}

In our case of weighted graph, we adapt Ollivier's definition of discrete Ricci curvature and relate it to community detection. On a weighted graph $(V, E, w)$ we consider the associated distance function $d$ defined by Equation~(\ref{dis}). To define the Ollivier Ricci curvature, one needs probability measures associated with each vertex. 
In \cite{Lin2011-wk}, given a non-negative scalar $\alpha$, Lin \etal defined the probability measure $m^{\alpha}_x$ on $V$ associated to the vertex $x \in V$ to be
\begin{equation} \label{mass_lin}
m^{\alpha}_x(i) = \left\{
        \begin{array}{lll}
            \alpha & \quad \mbox{\ if } i=x \\
            (1-\alpha)/|\pi(x)| & \quad \mbox{\ if } i \sim x \\
            0  & \quad  \mbox{\ otherwise},
        \end{array}
    \right.
\end{equation}
where $\pi(x)$ is the set of neighbors of $x$ and $|\pi(x)|$ is the number of elements in set. In this work, we extended and generalized the probability distribution defined by Lin \etal to further consider the edge weights. The probability measure we constructed is motivated by the classical differential geometry as follows. Consider a network as a discretization of a smooth manifold, a Riemannian metric and a Riemannian distance should correspond to edge weight and distance. 
Given an edge weighted graph $(V,E, w)$ whose associated distance is $d$ (Equation~(\ref{dis})), we define the associated normalized probability measure (Equation~(\ref{mass1})) as follows. 
% \begin{equation} \label{mass2}
% m^{\alpha}_x(i) = \left\{
%         \begin{array}{lll}
%             \frac{e^{-\alpha \eta_{x i}}}{\sum_{j \sim x} e^{-\alpha \eta_{j x}}} & \quad
%             i \sim x \\ 0  & \quad  others
%         \end{array}
%     \right.
% \end{equation}  

\begin{equation}\label{mass2}
m^{\alpha,p}_x(i)=
\begin{cases}
\alpha & \mbox{\ if } i = x\\
\frac{1-\alpha}{C}\cdot \exp(-(d(x,i))^p) & \mbox{\ if } i \sim x\\
0 & \mbox{\ otherwise},
\end{cases}
\end{equation}
where $C$ is a normalization factor $C= \sum_{j\sim x} \exp(-(d(x,j))^p)$. 
Then the discrete Ricci curvature defined under this mass distribution is as follows:
\begin{equation} \label{ricci2} \kappa_{xy}
=1-\frac{W(m^{\alpha,p}_x, m^{\alpha,p}_y)}{d(x,y)} 
\end{equation}
Note that for $p=0$, the probability measure becomes $m^{\alpha,0}_x(i)=\frac{1-\alpha}{|\pi(x)|}$ if $i \sim x$, which is the same as Equation (\ref{mass_lin}).  For our computation, we take $p=2$ and $\alpha=1/2$ in most cases.
% \saiba{massive change above, please double check. Checked, Feng}

To detect the community structure, we introduce a curvature guided diffusion process called discrete Ricci flow (or simply Ricci flow) on a network. The flow was motivated by the powerful tool of smooth Ricci flow which has recently revolutionized the geometry and topology of $3$-dimensional spaces. It is also related to the work of discrete Ricci flow on surfaces \cite{chow2003combinatorial}.
%\saiba{more surface Ricci flow citation here? No need, Feng}

%Starting with a connected weighted graph $(V,E, \eta)$ with edge weight $\eta$. 
The Ricci flow on weighted graph $(V,E, w^{(0)})$ generates a time dependent family of weighted graphs $(V,E, w(t))$ such that the initial value $w(0)=w^{(0)}$ and the weight $w_{ij}(t)$ on edge $ij$ changes proportional to the Ollivier-Ricci curvature $\kappa_{ij}(t)$ at edge $ij$ at time $t$.

Ollivier, in the future work section of~\cite{Ollivier2009-yl}, suggested to use the following formula for Ricci flow with continuous time parameter $t
$:
\begin{equation}\label{dc}
\frac{ d }{dt} d_{ij}(t)=-\kappa_{ij}(t) d_{ij}(t).
\end{equation}
%Combining Equation~(\ref{dc}),  Equation~(\ref{ricci}) and Equation ~(\ref{dis}), 
In our setting, we use discrete time parameter $k$ and define a discrete Ricci flow as a family of edge weighted graphs $(V, E, w^{(k)})$, $ k \in \mathbf Z_{\geq 0}$, with associated distances $d^{(k)}$, where $d^{(k)}(i, j)$ is the shortest edge path length between $i$ and $j$.
The flow takes discrete time $k\in \mathbf Z_{\geq 0}$. 
\begin{equation}\label{drf}
    w_{ij}^{(k+1)}= (1-\kappa^{(k)}_{ij}) d^{(k)}(i,j)  \quad  %w_{ij}(0)=\eta_{ij},
\end{equation}
where $\kappa_{ij}^{(k)}$ is the Ricci curvature at the edge $ij$ of the weighted graph
$(V,E, w^{(k)})$. Note that Equation~(\ref{drf})  states that  
$w_{ij}^{(k+1)}$ is equal to the Wasserstein distance between adjacent vertices $i$ and $j$ in the metric graph $(V, E, w^{(k)})$.
%W_{ij}(t-1)$ is the Wasserstein distance between the two adjacent vertices $i$ and $j$ in the metric $d$ associated to edge weight $w(t-1)$ (Equation ~(\ref{dis})). 
To detect a community structure on an unweighted graph $G=(V, E)$, we take the initial edge weight $w^{0)}$ to be the constant $1$ and run the Ricci flow (Equation~(\ref{drf})) from there.

The flow tends to expand negatively Ricci curved subgraphs and contract positively Ricci curved subgraphs.
In networks with clear community structures, the inter-community edge weights converge to $\infty$ and the intra-community edge weights converge to $0$. Thus the network is naturally partitioned into different communities by removing  edges of very high weights.  
%We obtain a new graph $(V, E', \eta')$.  The weight $\eta'$ is the inherited weight $w(T)$ at the blow-up time $T$. Then we start the discrete Ricci flow on the new graph.
We may continue to run Ricci flow to further partition the communities into smaller ones, when the network has  hierarchical community structures.  This is similar to the popular Girvan-Newman algorithm~\cite{Girvan2002-cq}. It first computes betweenness centrality for each edge, then removes edges with the highest score. After that, re-compute all scores and repeat. Here the betweenness of an edge is the sum $\sum_{i, j \in V\{\partial e\}} \frac{ \sigma_{ij}(e)}{\sigma_{ij}}$ where $\sigma_{ij}$ is the number of shortest paths from $i$ to $j$ and $\sigma_{ij}(e)$ is the number of shortest paths from $i$ to $j$ which contain the edge $e$. Since computing centrality involves global information of the network and is expensive, Ricci flow is easier to implement in practice.

Our work on Ricci curvature on networks builds on our previous work~\cite{Ni2015-yv} and is also inspired by the important works of E. Saucan and J. Jost \etal in \cite{Samal2018-bt, Saucan2018-yg, Sreejith2017-md, Jost2014-rk}. In these works, they systematically introduced and investigated various discrete curvatures for complex networks. The comparative analysis of Forman and Ollivier Ricci curvature on benchmark dataset of complex networks and real-world networks was also carried out. Their numerical results show a striking fact that these two completely different discretization of the Ricci curvatures are highly correlated in many networks.

\section{From Ricci flow to community detection}

In this section, we explain the relationship between the discrete Ricci flow (Equation~\ref{drf}) on networks and the classical theory of 3-manifolds and the Ricci flow. These classical theories motivate us to apply Ricci flow for community detection.

A 3-manifold in mathematics is a connected (Hausdorff) topological space which looks like the 3-space in small scale. More precisely, each point in a 3-manifold has a neighbourhood homeomorphic to an open ball in the 3-space.  3-manifolds are the basic prototypes of 3-dimensional spaces. There are two operations which produce complicated 3-manifolds from simple ones. The first is the connected sum operation and the second is the torus sum.  In the connected sum operation, one takes two 3-manifolds, removes two small open balls from them. Glue the two remaining 3-manifolds with holes along their 2-sphere boundary by a homeomorphism. The resulting 3-manifold is the connected sum of the given two. The torus sum operation is similar where one takes two 3-manifolds with torus boundary and glues the two boundary by a homeomorphism. Naturally, one asks if each 3-manifold can be decomposed as a connected sum and torus sum from simple pieces, i.e., reversing the above process. The answer is affirmative and is the content of the classical theorem of topological decomposition of 3-manifolds, established  by H. Kneser, J. Milnor ~\cite{Milnor1962-pm} Jaco-Shalen~\cite{Jaco1979-pw}, and Johannson~\cite{Johannson2006-vx}.
The theorem states that each 3-manifold can be canonically decomposed into simple pieces using the 2-spheres and tori. In the sphere decomposition, one looks for topologically essential disjoint 2-spheres in the 3-manifold. Here a topologically essential 2-sphere (or torus) means it is not the boundary of a 3-ball (or a solid torus) in the manifold. By capping off the 2-sphere boundary of the compliment by the 3-balls, one obtains new 3-manifolds which have simpler topology and the original 3-manifold is obtained from them by the connected sum. The torus decomposition is similar in construction and uses the essential torus instead of the 2-sphere. A 3-manifold which cannot be decomposed by any essential 2-spheres and tori is called atoroidal.
Thus the classical decomposition theorem says that each 3-manifold is a connected sum and torus sum of a collection of atoroidal ones.
%In our work,  we take these  atoroidal pieces as communities in the manifold.

One of the most important problems in low-dimensional geometry and topology is the geometrization conjecture proposed by William Thurston in 1976 \cite{Thurston1982-rf}. It states that any atoroidal 3-manifold admits complete, locally homogeneous Riemannian metrics, i.e., atoroidal 3-manifolds are geometric. This fundamental conjecture was solved by the groundbreaking work of Perelman in 2004 which uses the Ricci flow method developed by R. Hamilton in 1981.  A key step in Perelman's proof is that the Ricci flow can detect the 2-sphere and torus decomposition.

The Ricci flow, introduced by Richard S. Hamilton in 1981 \cite{hamilton1982three}, deforms the metric of a Riemannian manifold in a way formally analogous to the diffusion of heat, smoothing out irregularities in the metric. Ricci flow has been one of the most powerful tools for solving geometric problems in the past forty years. The flow exhibits many similarities with the heat equation.  

Suppose a Riemannian metric $g_{ij}$ is given on a manifold $M$ so that its Ricci curvature is $R_{ij}$. Then Hamilton's Ricci flow is the following second-order nonlinear partial differential equation on symmetric $(0, 2)$-tensors:
$$\frac{\partial}{\partial t}g_{ij}=-2 R_{ij}.$$
A solution to the Ricci flow is a one-parameter family of metrics $g_{ij}(t)$ on a smooth manifold $M$ satisfying the above partial differential equation. One of the key properties of the Ricci flow is that the curvature evolves according to a nonlinear version of heat equation. Thus the Ricci flow tends to smoothing out irregularity of the curvature. % A pictorial illustration of what the Ricci flow does to 3-manifolds is shown in Fig.~\ref{fig:manifold}.

% \begin{figure}[htbp]
%     \centering
%     \subfigure[Ricci flow on a manifold]{
%         \label{fig:manifold:subfig:manifold}
%         \includegraphics[width=0.3\columnwidth]{figures/rf_manifold}}
%     \subfigure[A graph of four communities]{
%         \label{fig:manifold:subfig:uncut}
%         \includegraphics[width=0.3\columnwidth]{figures/rf_graph_uncut}}
%     \subfigure[The graph after Ricci flow]{
%         \label{fig:manifold:subfig:cut}
%         \includegraphics[width=0.3\columnwidth]{figures/rf_graph_cut}}
%     \caption{An illustration of Ricci flow on a manifold and a graph.}
%     \label{fig:manifold} %% label for entire figure
% \end{figure}

In the groundbreaking work of G. Perelman, he used the Ricci flow to prove the geometrization conjecture by analyzing the singularity formation and introducing a Ricci flow with `surgery' \cite{mt}. One way a singularity may arise in the Ricci flow is that an essential 2-sphere in the manifold may collapse to a point in finite time. Also essential tori in the manifold can also be detected by the Ricci flow. Thus, one consequence of Perelman's work is that Ricci flow can be used to find the 2-sphere and torus decomposition in a 3-manifold. This is the key motivation for us to introduce a discrete Ricci flow on networks for community detection.

We consider a community structure in a network heuristically  to be a discrete counterpart of the (topological) sphere-torus decomposition of a 3-manifold. One can roughly justify this analogy as follows.
A commonly agreed characterization of communities in a network is that there are more edges linking nodes  within a community and fewer edges linking nodes in different community.  This is similar to the components in the sphere decomposition. Indeed,  according to the Seifert-van Kampen theorem (\cite{spanier1966algebraic}), loops within a component in the sphere and torus decomposition interact more among themselves than loops in different components.
%\textcolor{red}{F: Jie, is this OK to say the following sentence?}
Since the Ricci flow runs on a manifold with a Riemannian metric and detects topological decompositions, the discrete Ricci flow on a network with edge weight should detect the community structures.  This theme has been supported by our numerical experiments.

\section{Algorithms for discrete Ricci flow on graph}\label{sec:algorithm}

For a given graph $G=(V,E)$, let the edge weight of edge $xy \in E$ be $w_{x y}$, the discrete Ricci curvature of edge $xy$, $\kappa_{xy}$, is computed as follows:
$$ \kappa_{xy}=1-\frac{W(m^{\alpha,p}_x, m^{\alpha,p}_y)}{d(x,y)}, $$
where $W$ is the optimal mass transport distance (a.k.a. Wasserstein Distance or Earth Mover Distance). The mass distribution $m^{\alpha,p}$ is defined as 
$$
m^{\alpha,p}_x(i)=
\begin{cases}
\alpha & \mbox{\ if } i = x\\
\frac{1-\alpha}{C}\cdot \exp(-(d(x, i))^p) & \mbox{\ if } i \sim x\\
0 & \mbox{\ otherwise},
\end{cases}
$$
where $C$ is a normalization factor $C= \sum_{j\sim x} \exp(-(d(x ,j ))^p)$. To speed up the computation of $W$, we suggest to apply the Sinkhorn distance~\cite{Cuturi2013-kx} which computes the approximate optimal transport distance fast without losing too much accuracy.

For discrete Ricci flow computation, we follow Equation~\ref{drf}. For each iteration $i$, the Ricci flow metric on edge (edge weight) is defined as follows:
$$ w^{(i+1)}_{xy}=d^{(i)}(x,y)-\eps \cdot \kappa^{(i)}_{xy} \cdot d^{(i)}(x,y) ,\quad \forall xy \in E,$$
where $\eps$ is the step size or learning rate of the gradient decent process. $d^{(i+1)}(x,y)$ is the shortest path between $x$ and $y$ based on $w^{(i+1)}$.

Since the edge weights change in a relative manner, we re-scale the metric so that the sum of edge weight remains constant. For each iteration, we recompute the Ricci curvature on each edge based on the current edge weight, then update and normalize the weight until $|\kappa^{(i)}_{xy}-\kappa^{(i-1)}_{xy}|<\delta$, where $\delta >0$.
The detailed algorithm is as follows.

\begin{algorithm}
    \label{algo:ricci_flow}
    \SetKwInOut{Input}{Input}
    \SetKwInOut{Output}{Output}

    \Input{An undirected graph $G$ and a real number $\delta > 0$.}
    \Output{A weighted graph $G$ with edge weight as Ricci flow metric on each edge.}

    Normalize the edge weight $w^{(i)}_{xy} \leftarrow d^{(i)}(x,y) \cdot \frac{|E|}{\sum_{xy \in E}d^{(i)}(x,y)}$

    Compute the Ricci curvature of $G$

    Update the edge weight by $w^{(i+1)}_{xy}\leftarrow d^{(i)}(x,y)-\eps \cdot \kappa^{(i)}_{xy} \cdot d^{(i)}(x,y)$

    Repeat $1-3$ until all $|\kappa^{(i)}_{xy}-\kappa^{(i-1)}_{xy}|<\delta$.

    \caption{Discrete Ricci Flow}
\end{algorithm}

\section{Evaluations}\label{sec:more_evaluation}

Here we briefly introduce these two metrics and apply these metrics to measure the accuracy of our clustering results. 

\paragraph{ARI} \label{par:ari}
Adjusted Rand Index (ARI)~\cite{hubert1985comparing} measures the similarity between two clustering results by adjusting the Rand index~\cite{rand1971objective} in a way that a random result gets a score of $0$. Suppose the ground truth community is described by a partition of nodes $\{1, 2, \cdots, n\}$ into disjoint sets $C_1, C_2, \cdots, C_m$, and the detected community is represented by disjoint sets $\{\bar{C}_1, \bar{C}_2, \cdots, \bar{C}_k\}$. ARI scores the agreement of partitioned node pairs in $C_{i}$ and $\bar{C}_{j}$. The ARI of the given clustering result and the actual result is defined as follows. 

\begin{equation}\label{ari}
    ARI(C,\bar{C}) = \frac{\sum_{i=1}^{m}\sum_{j=1}^{k}\binom{|C_{i}\cap\bar{C}_{j}|}{2}-[\sum_{i=1}^{m}\binom{|C_{i}|}{2}\sum_{j=1}^{k}\binom{|\bar{C}_{j}|}{2}]/\binom{|V|}{2}}{\frac{1}{2}[\sum_{i=1}^{m}\binom{|C_{i}|}{2}+\sum_{j=1}^{k}\binom{|\bar{C}_{j}|}{2}]-[\sum_{i=1}^{m}\binom{|C_{i}|}{2}\sum_{j=1}^{k}\binom{|\bar{C}_{j}|}{2}]/\binom{|V|}{2}}
\end{equation}
The value of ARI is $1$ when the two clustering results match perfectly and is below $1$ otherwise.

% \paragraph{NMI} \label{par:nmi}

% Normalized Mutual Information (NMI)~\cite{vinh2010information} is a measure borrowed from information theory and also has been widely used to compare community detection algorithms.
% Let $G = (V,E)$ be a graph with community structures, then $C = \{C_{1}, C_{2}, \dots, C_{m}\}$ is the true clustering of $G$ and $\bar{C} = \{\bar{C}_{1}, \bar{C}_{2}, \dots, \bar{C}_{k}\}$ is the clustering delivered by a community detection algorithm. The entropy of $C$ is

% \begin{equation*}
%     H(C) = -\sum_{i=1}^{m}\frac{|C_i|}{|V|}\log\frac{|C_i|}{|V|}.
% \end{equation*}

% The NMI is defined by
% \begin{equation*}
%     NMI(C,\bar{C}) = \frac{\sum_{i=1}^{m}\sum_{j=1}^{k}\frac{|C_{i}\cap \bar{C}_{j}|}{|V|}\log\frac{|V| |C_{i}\cap \bar{C}_{j}|}{|C_{i}||\bar{C}_{j}|}}{\sqrt{H(C)H(\bar{C})}}.
% \end{equation*}
% If $C$ is identical to $\bar{C}$, $NMI(C,\bar{C})$ equals $1$. 
% %If the similarity between them is equal or less than what is expected from two random clusterings, $NMI(C,\bar{C})$ becomes $0$.

\paragraph{Modularity}

Modularity introduced in \cite{newman2004finding} describes the fraction of edges between communities minus the expected fraction if edges are distributed uniformly at random. The range of modularity is $[-1,1)$ where a positive value means more edges are intra-community edges than expected for a random graph. Unlike ARI, Modularity does not require ground truth community, and is easy to compute even for large networks. But modularity suffers from resolution problem due to the statistical nature~\cite{Fortunato2007-ff}.

For a network $G$ with $n$ node, $m$ edges, $c$ communities and adjacency matrix $A$, the modularity $Q$ of this network is defined as follows:
$$
Q = \frac{1}{(2m)}\sum_{vw} \left[ A_{vw} - \frac{k_v k_w}{(2m)} \right] \delta(c_{v}, c_{w})     
  = \sum_{i=1}^{c} (e_{ii}-a_{i}^2),
$$
where $k_x$ is the degree of node $x$, $\delta(c_{v}, c_{w})$ equals to $1$ if node $v$ and $w$ are in the same community, otherwise $0$, and $e_{ij}$ is the fraction of inter-community edges connecting community $i$ and $j$:
$$
e_{ij}= \sum_{vw} \frac{A_{vw}}{2m} 1_{v\in c_i} 1_{w\in c_j},
$$
and $a_i$ is the fraction of intra-community edges in community $i$:
$$
a_i=\frac{k_i}{2m}
   = \sum_{j} e_{ij}.
$$

\subsection{Edge Weights under Ricci Flow Iterations}
We first analyze the iterating process of Ricci flow on graphs with community structure. In Fig.~\ref{fig:iteration}, we showed the variation of edge weights and edge Ricci curvatures on an LFR graph with $1000$ nodes and $\mu=0.4$ with $38$ communities and with average degree $20$. We name edges to connect nodes in the same community as intra-community edges (labeled as blue) and edges connect nodes in two different communities as inter-community edges (labeled red in the figure). When the iteration process starts, the original weight of all edges is set to be $1$ (Fig.~\ref{fig:iteration:subfig:weight}). As Ricci flow iterates, we can see a clear trend that the weights of intra-community edges (labeled blue in figure) are decreasing and converged to $0$ while the weights inter-community edges increased with iterations. In Fig.~\ref{fig:iteration:subfig:rc}, as Ricci flow iterates, the Ricci curvature on all edges converged to a fix value.

\begin{figure}[htbp]
    \centering
    \subfigure[Ricci flow iterations over edge weights]{
        \label{fig:iteration:subfig:weight}
        \includegraphics[width=0.48\columnwidth]{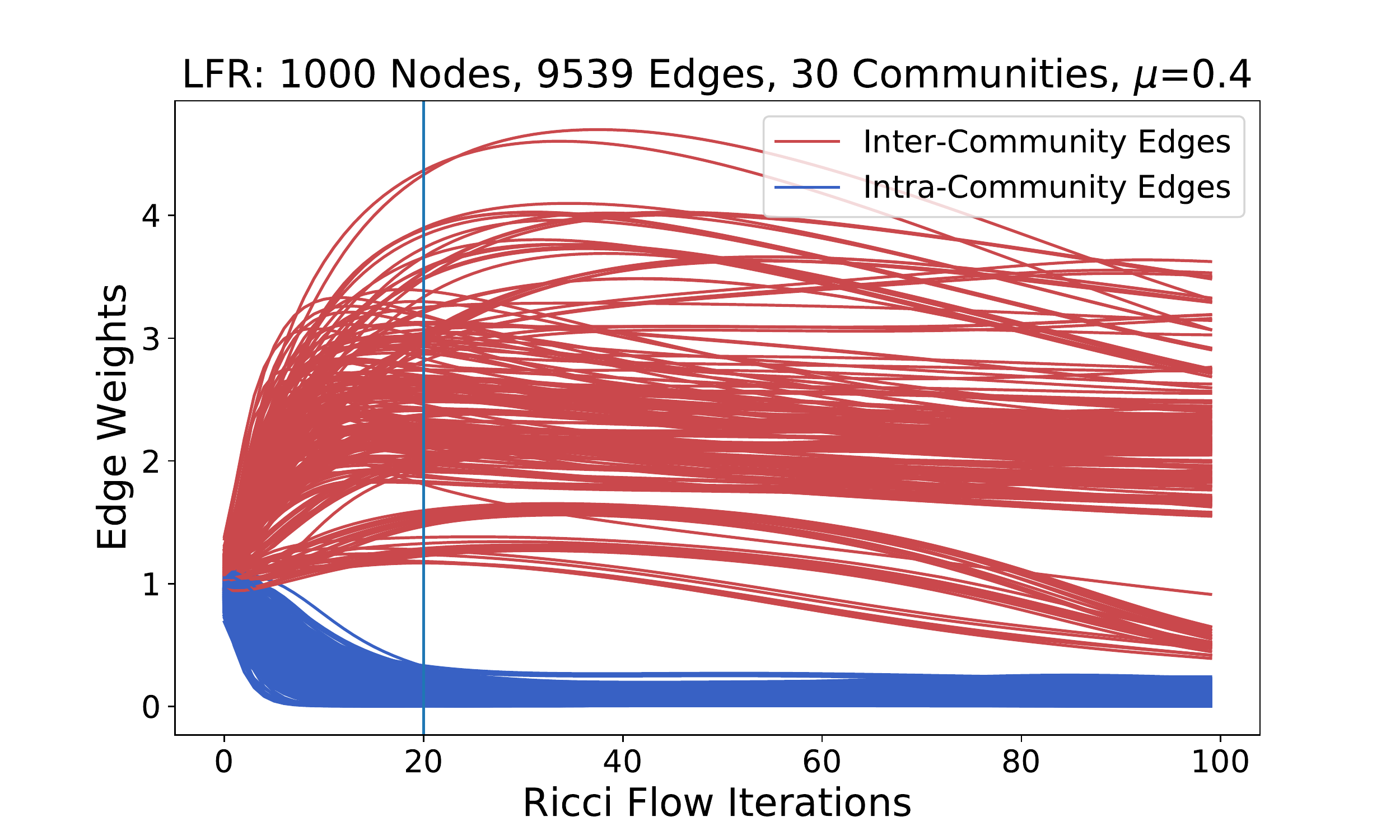}}
    \subfigure[Ricci flow iterations over edge Ricci curvature]{
        \label{fig:iteration:subfig:rc}
        \includegraphics[width=0.48\columnwidth]{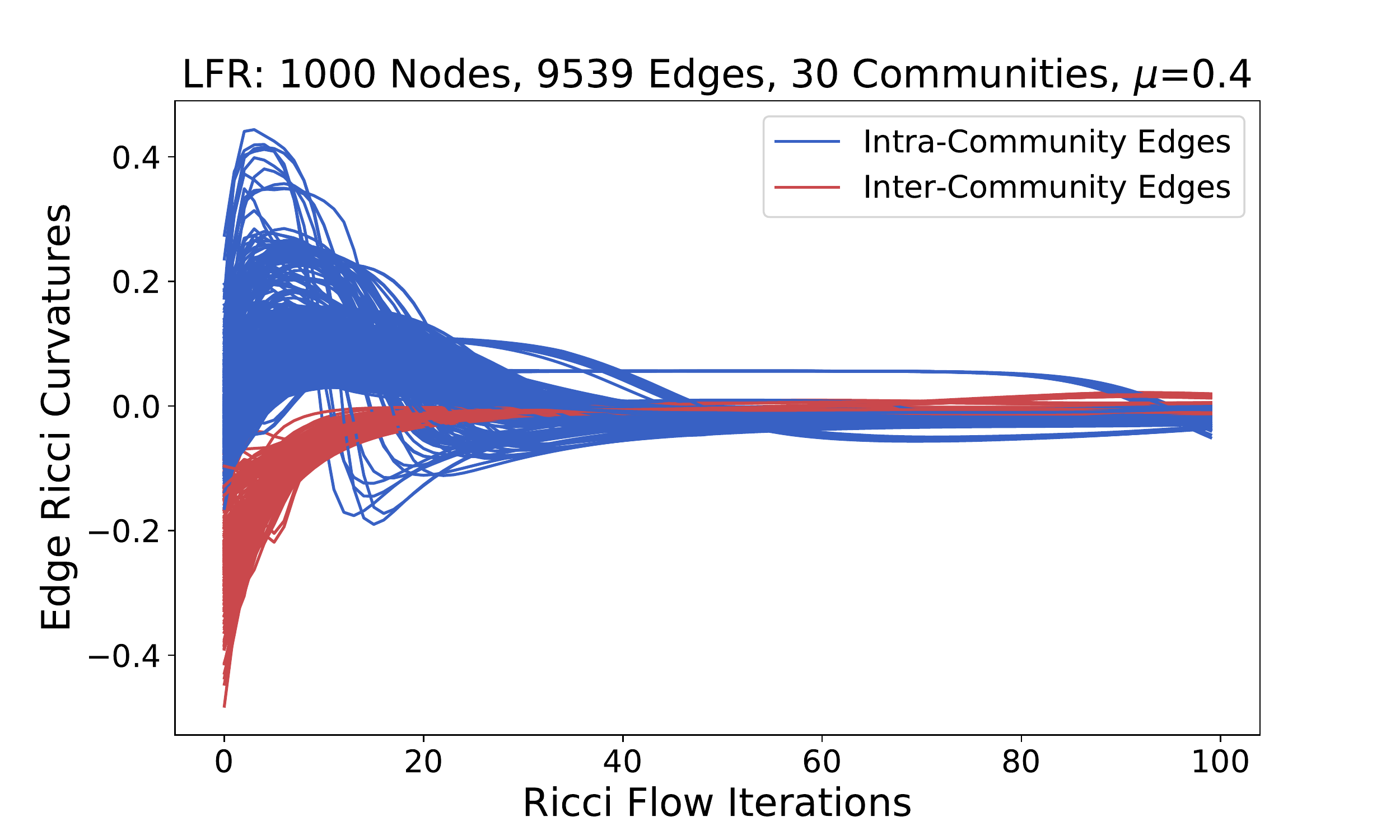}}
    \caption{ Edge weights and Ricci curvatures over Ricci flow iterations on a LFR graph with $1000$ nodes, $9539$ edges, and $\mu=0.4$. Fig.~\ref{fig:iteration:subfig:weight}: the edge weights of inter-community edges expand by Ricci flow. Fig.~\ref{fig:iteration:subfig:rc}: edge Ricci curvatures converged to $0$ by Ricci flow. }
    \label{fig:iteration} %% label for entire figure
\end{figure}

\subsection{Ricci Flow Parameters}
Here we discuss the influence of different probability measures on Ricci flow. Recall that, when computing the Ricci curvature of each edge, we choose the probability distribution as exponential in the edge weight. This exponential distribution takes \emph{base} to be $e$ and the exponent to be the edge weight to the \emph{power} of $p$. In Fig.~\ref{fig:power}, we test the community detection accuracy of Ricci flow using probability distribution with different base and power on LFR graph with $1000$ nodes with $\mu=0.5$ and average degree $20$. The result shows that over iterations, choosing $p=2$ or $p=3$ yields decent accuracy. Notice that for $p=0$, the probability distribution is reduced to Lin and Yau's setting~\cite{Lin2011-wk} that the mass is equally distributed to a node's neighbors during the optimal transportation process.

\begin{figure}[htbp]
    \centering
    \includegraphics[width=0.5\columnwidth]{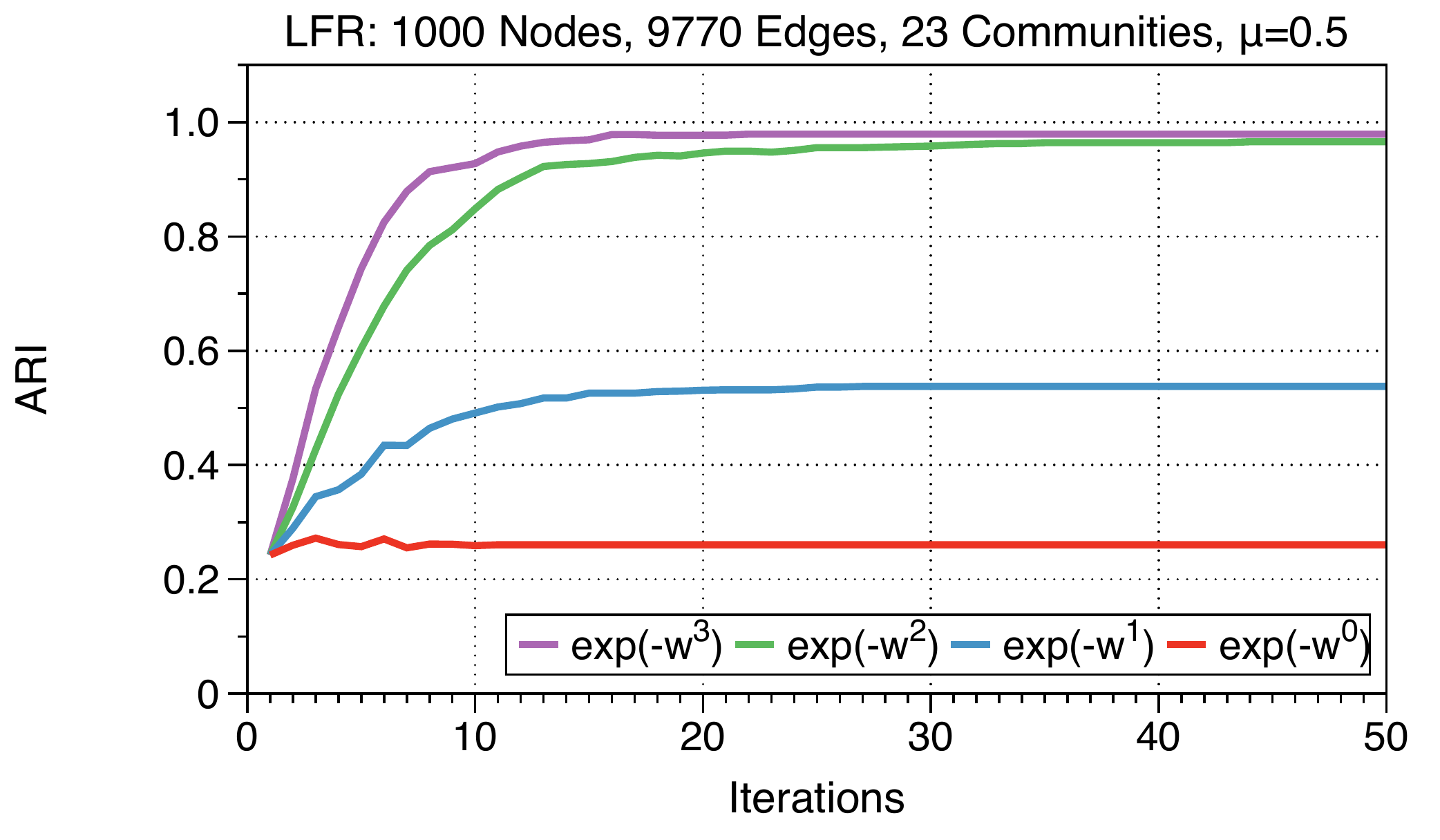}
    \caption{ Ricci flow with different power probability measure setting on a LFR graph with $1000$ nodes, 9770 edges and $\mu=0.5$.}
    \label{fig:power} %% label for entire figure
\end{figure}

In Fig.~\ref{fig:base}, we check the influence of probability distribution of different \emph{base} with $p=2$ on LFR graph with $1000$ nodes and average degree $20$. To eliminate the influence of randomness of the model, for every different $\mu$, we take the average result over $10$ trials. The result shows that for most of the \emph{base} setting, the result is good when $\mu$ is smaller than $0.6$. For a more mixed graph with greater $\mu$, choose \emph{base} to be $e$ yields the most stable result.

In the following experiments, we stick to $base=e$ and $power=2$ as our parameter setting.

\begin{figure}[htbp]
    \centering
    \includegraphics[width=0.45\columnwidth]{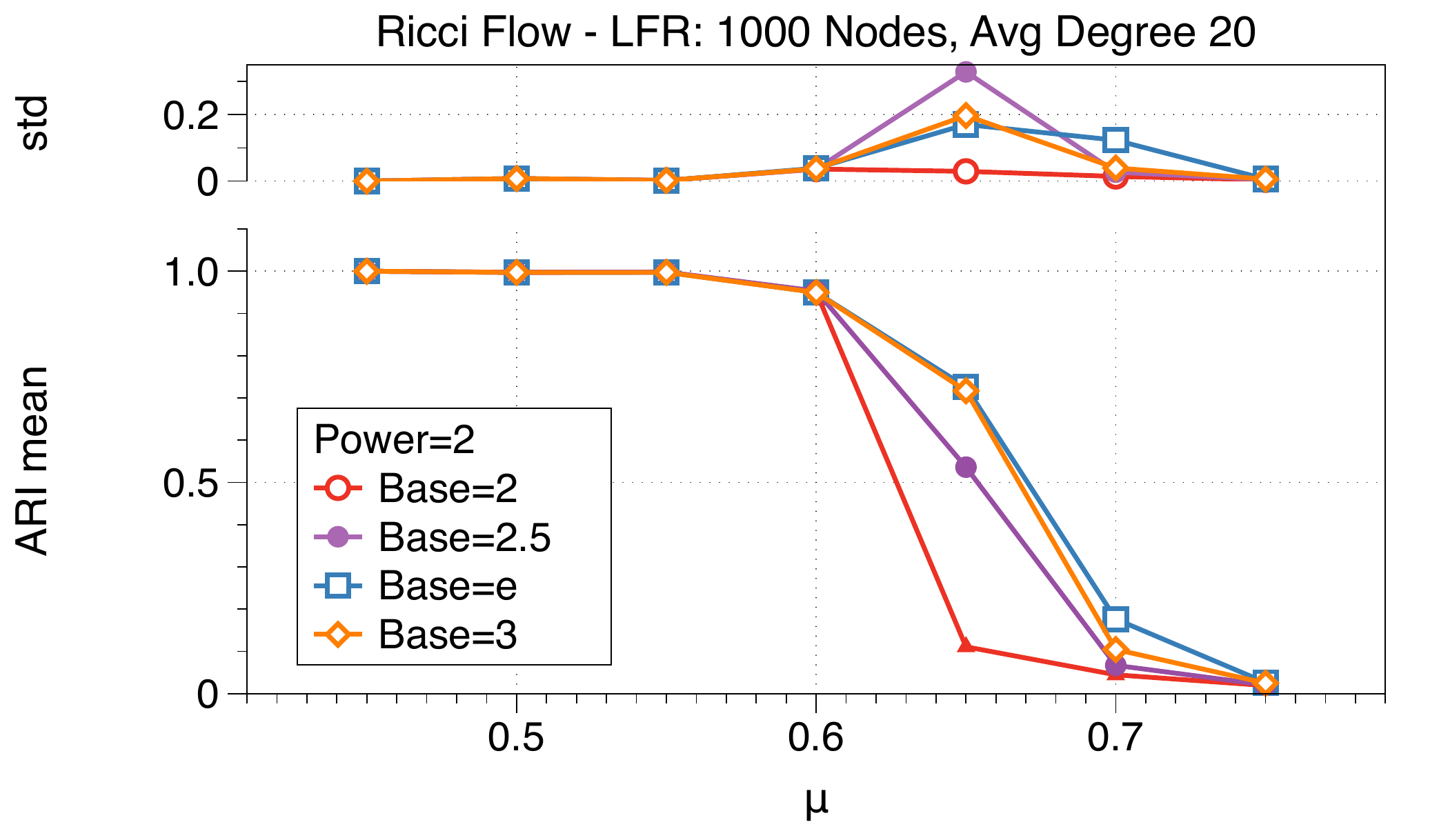}
    \caption{Ricci flow with different base probability measure setting on a LFR graph with $1000$ nodes, $9770$ edges and $\mu=0.5$.} 
    \label{fig:base} 
\end{figure}

% \begin{figure}[htbp]
%     \centering
%     \subfigure[]{
%         \label{fig:exp:subfig:ARI}
%         \includegraphics[width=0.45\columnwidth]{figures/ARI-exp}}
%     \subfigure[]{
%         \label{fig:exp:subfig:NMI}
%         \includegraphics[width=0.45\columnwidth]{figures/NMI-exp}}
%     \caption{ ARI and NMI on different exponential power.  }
%     \label{fig:exp} %% label for entire figure
% \end{figure}

\subsection{Approximated Optimal Transport Distance}
To speed up the computation process, we suggest utilizing the Sinkhorn distance~\cite{Cuturi2013-kx} as an approximate solution for optimal transport. The result of applying these two distances for Ricci flow is shown in Fig. \ref{fig:sinkhorn}. With regularity term set to be $0.1$, the Sinkhorn distance performs equally well as the optimal transport in Ricci flow process. And the time complexity is reduced by four times. The computation time is average over $5$ iterations on an Intel(R) Xeon(R) CPU E5-2670 v2 @ 2.50GHz with 512G RAM using 20 processes. Optimal transport is computed by CVXPY(https://www.cvxpy.org) with ECOS solver, Sinkhorn distance is solved by POT: Python Optimal Transport(https://github.com/rflamary/POT).

\begin{figure}[htbp]
    \centering
    \includegraphics[width=0.45\columnwidth]{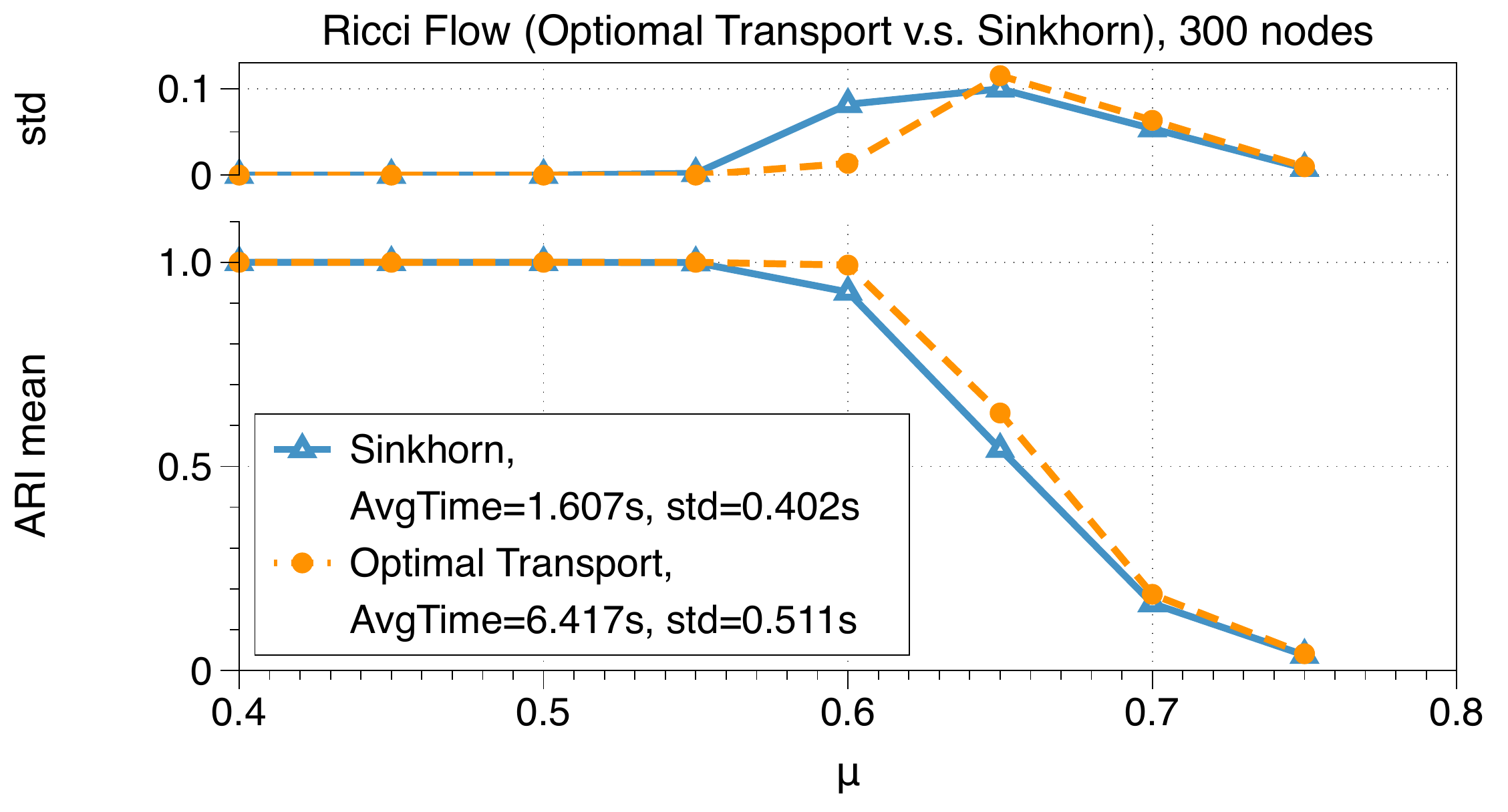}
    \caption{ A comparison of applying Sinkhorn distance and optimal transport distance for Ricci flow computations on LFR graph with $300$ nodes. The result of applying the Sinkhorn distance is similar as the one using optimal transport wile the time cost is four times smaller.} 
    \label{fig:sinkhorn} 
\end{figure}

\subsection{Ricci flow with Surgery}
Without too many parameter settings, Ricci flow yields good clustering performance for the graph with distinct community structures. In more complex graphs especially those with hierarchical community structures, we may need to run Ricci flow multiple times and do surgery during the Ricci flow iterations.
%the community structure is harder to capture since graphs come with lesser edges in the community and more edges between communities. Example of these weaker community graphs includes LFR graph with higher $\mu$, SBM with higher $p_{intra}/p_{inter}$ ratio, graphs with lesser edges in the community, and graphs with hierarchical structures.
For these kinds of graphs, we mimic the surgery process in the classic Ricci flow process during Ricci flow iterations. For every $5$ iteration in the Ricci flow process, we cutoff the edges whose weights rank among the top $5\%$. Fig.~\ref{fig:surgery-iteration} presents the edge weight and edge Ricci curvature changes during Ricci flow iterations with surgery. For every $5$ iteration, the surgery process cuts out a small part of high weighted edges (mostly inter-community edges), and the surgery helps to further separate the communities. Fig.~\ref{fig:surgery} demonstrates the result of Ricci flow with surgery. With LFR graph of $1000$ nodes and $\mu=0.6$, the surgery process boosts the ARI from $0.3$ to $0.8$.

% +add density

\begin{figure}[htbp]
    \centering
    \subfigure[Weight: Without Surgery]{
        \label{fig:surgery-iteration:subfig:weight_nosurgery}
        \includegraphics[width=0.48\columnwidth]{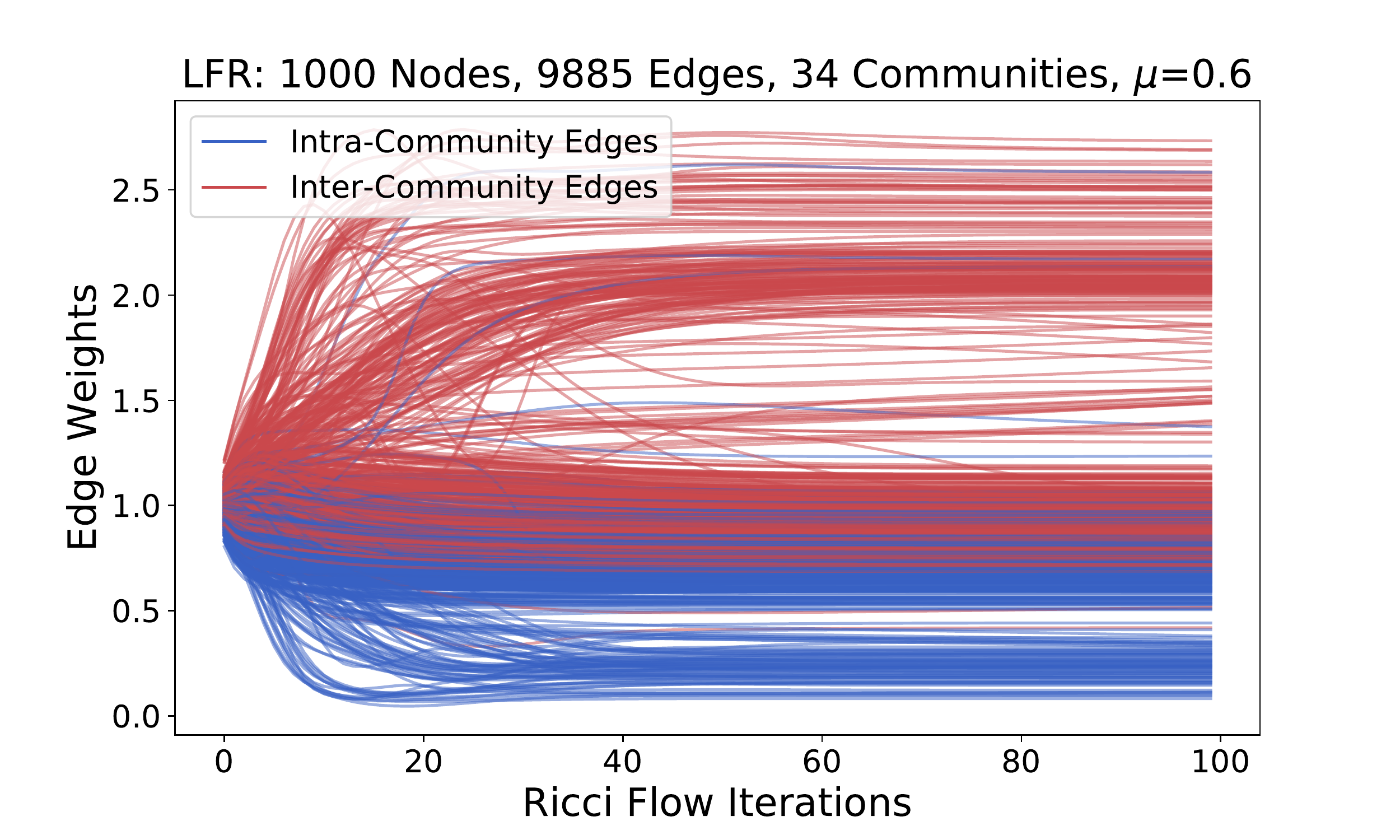}}
    \subfigure[Ricci Curvature: Without Surgery]{
        \label{fig:surgery-iteration:subfig:rc_nosurgery}
        \includegraphics[width=0.48\columnwidth]{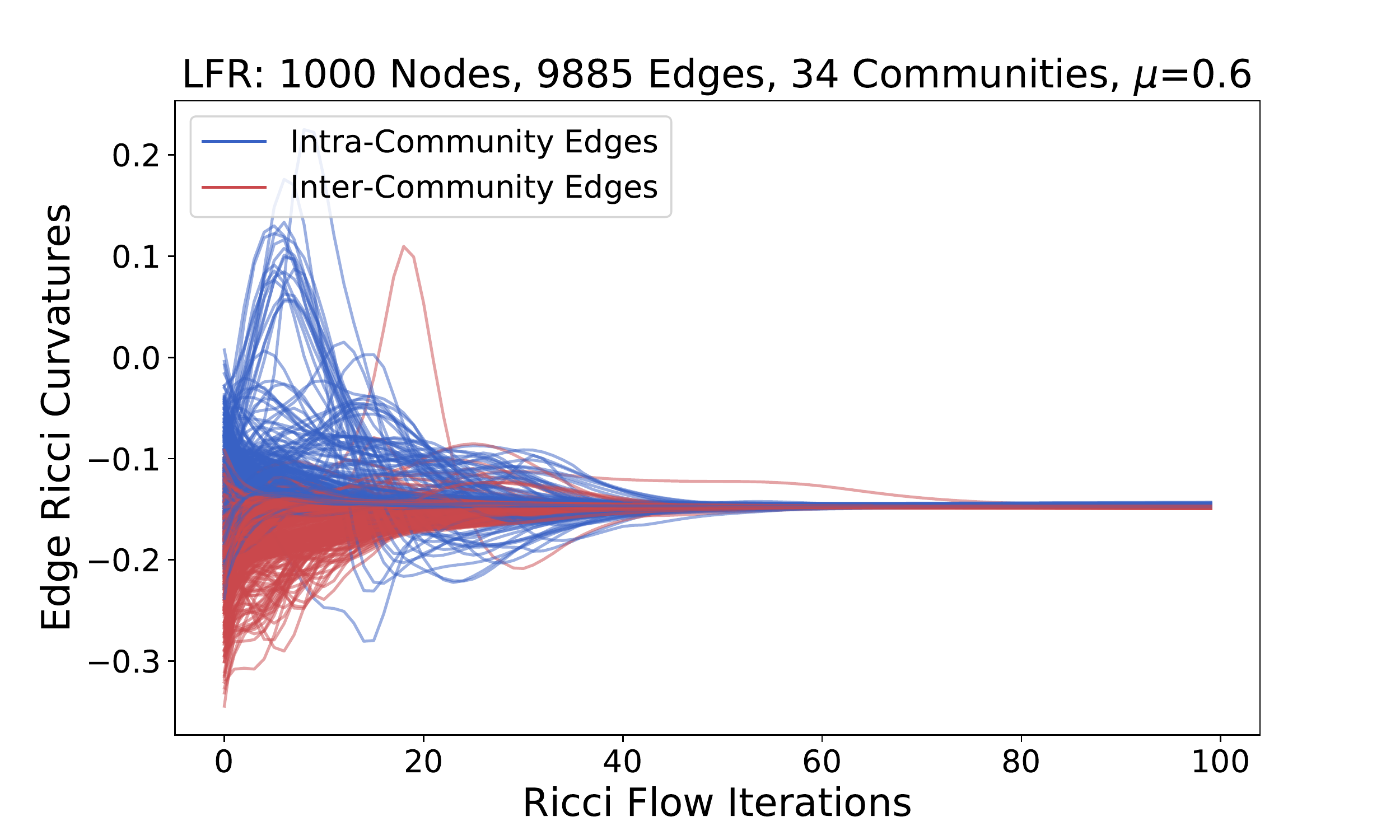}}
    \subfigure[Weight: With Surgery]{
        \label{fig:surgery-iteration:subfig:weight}
        \includegraphics[width=0.48\columnwidth]{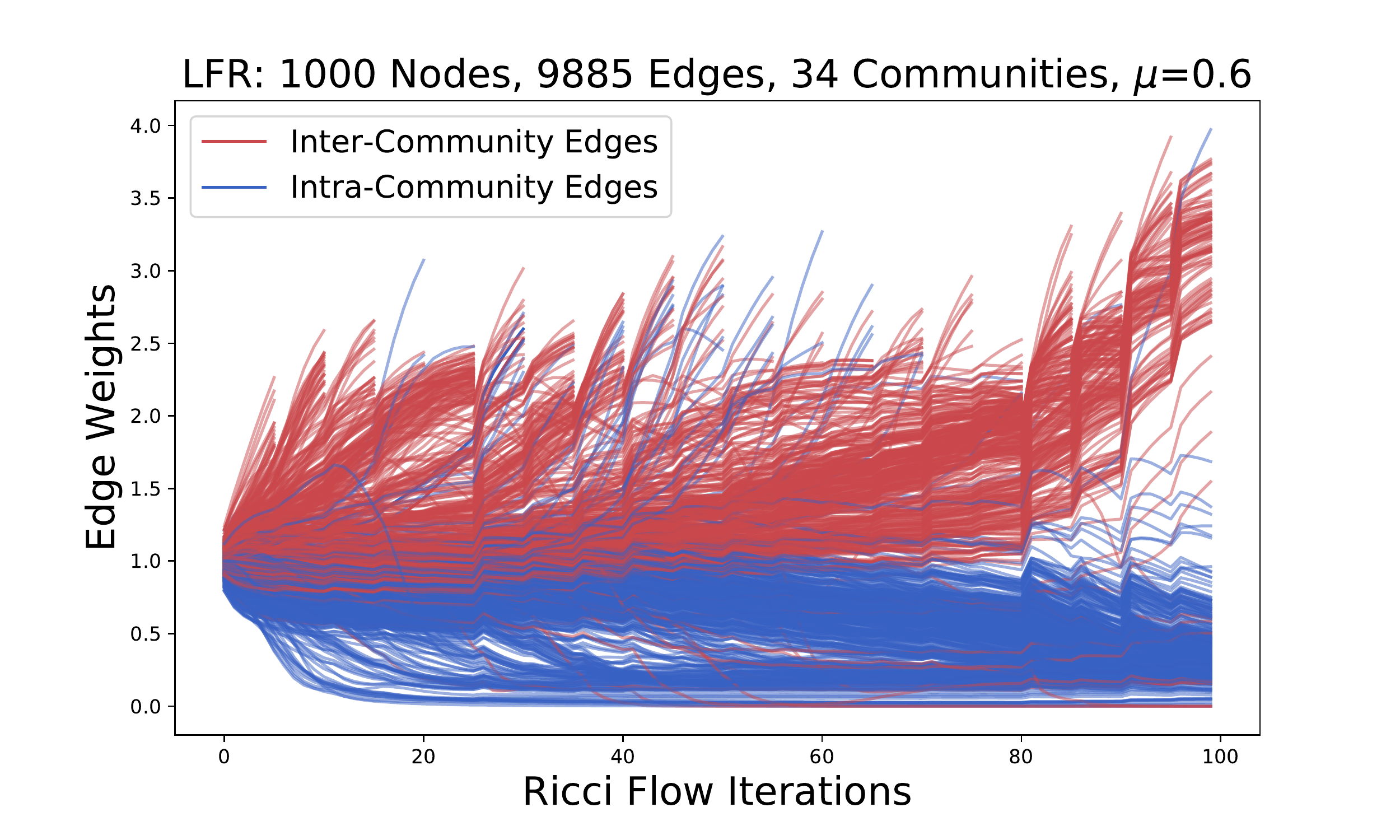}}
    \subfigure[Ricci Curvature: With Surgery]{
        \label{fig:surgery-iteration:subfig:rc}
        \includegraphics[width=0.48\columnwidth]{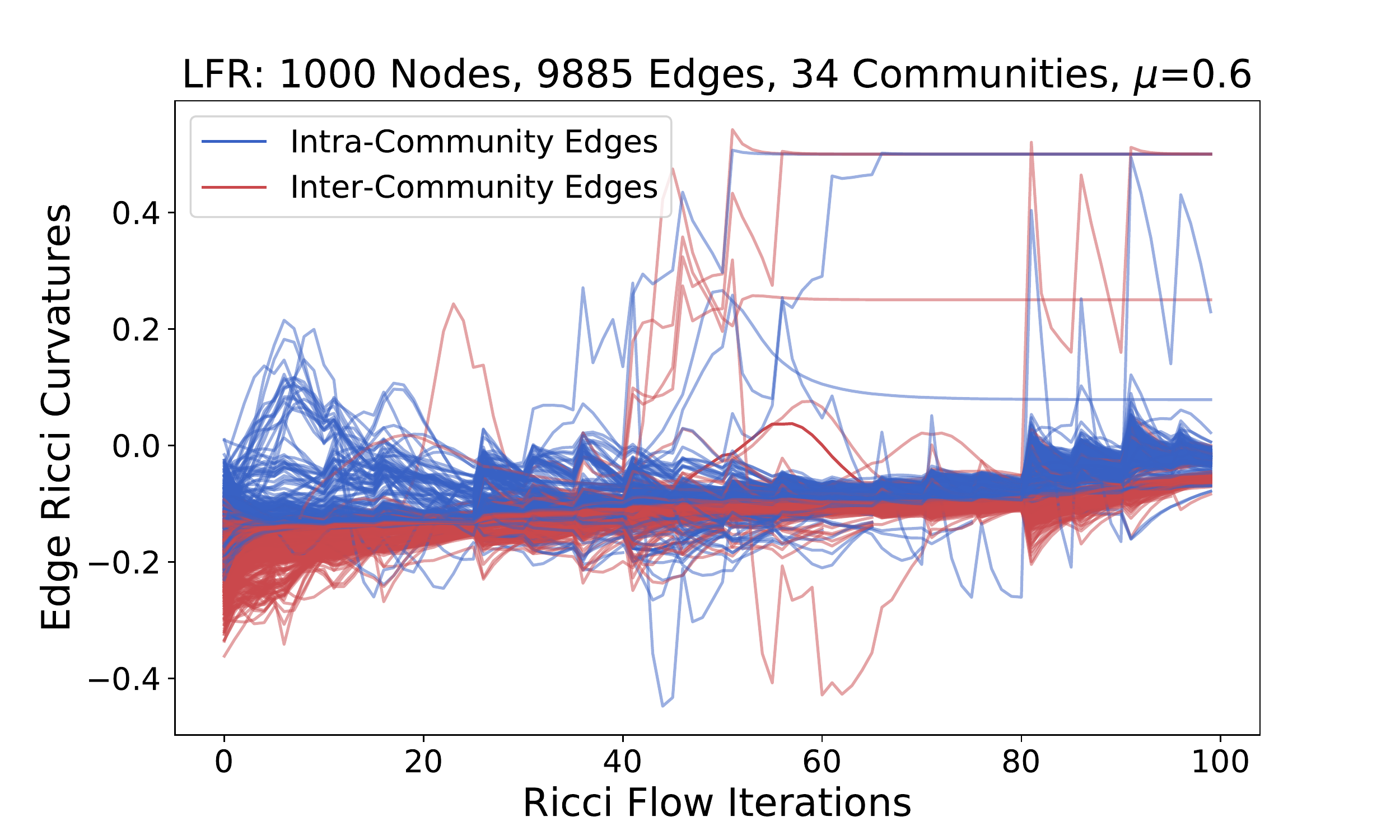}}
    \caption{ Edge weights and Ricci curvatures over Ricci flow iterations with and without surgery on a LFR graph with $1000$ nodes, $9718$ edges, and $\mu=0.6$. In Fig.~\ref{fig:surgery-iteration:subfig:weight} and Fig.~\ref{fig:surgery-iteration:subfig:rc}, for every $5$ iterations, the surgery process cut out a small part of high weighted edges (mostly inter-community edges). This cutting helps further separate the communities. }
    \label{fig:surgery-iteration} %% label for entire figure
\end{figure}

\begin{figure}[htbp]
    \centering
    \subfigure[Ricci Flow]{
        \label{fig:lfr-mu:subfig:rf_e2_ari}
        \includegraphics[width=0.45\columnwidth]{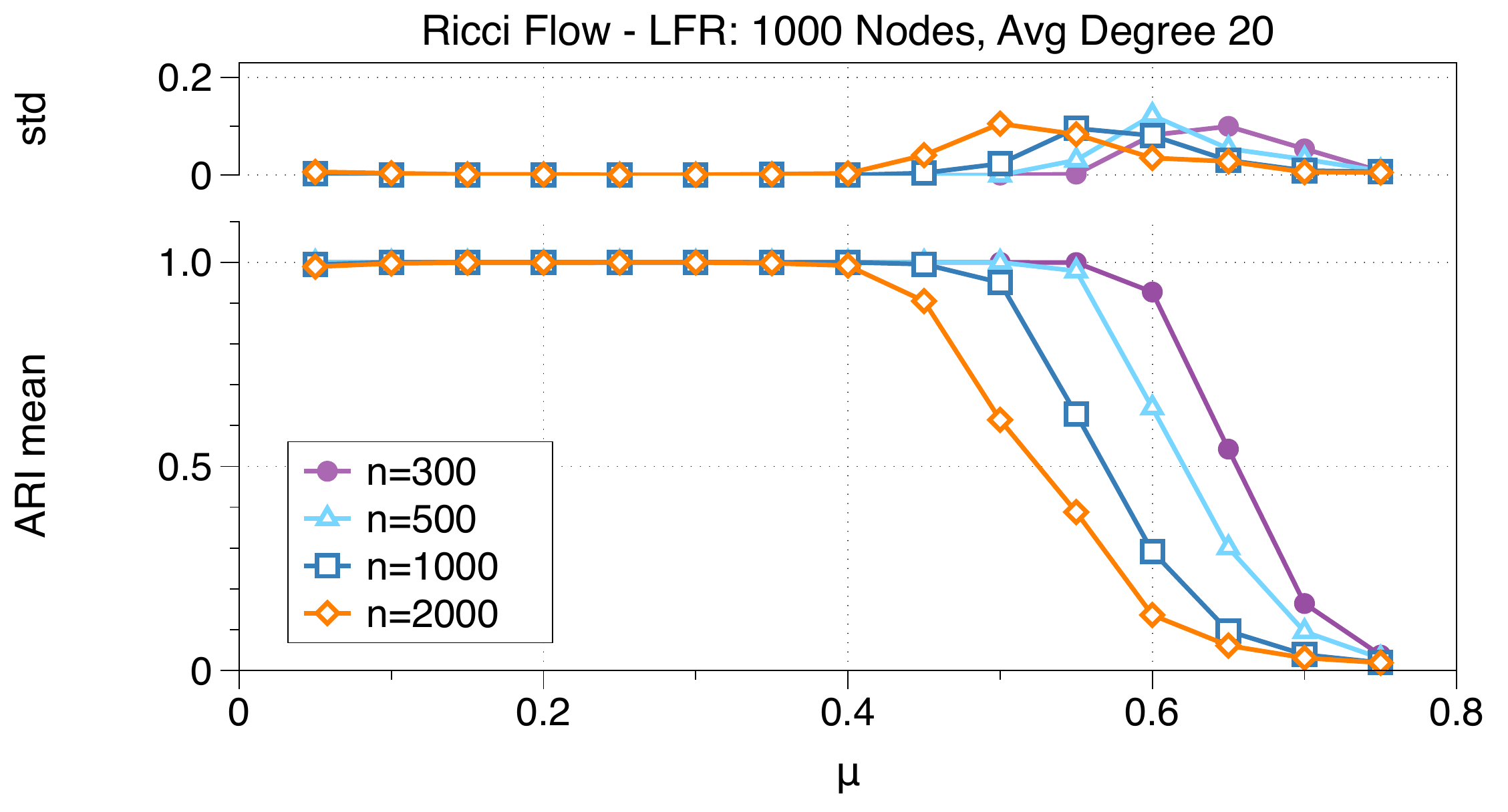}}
    \subfigure[Ricci Flow With Surgery]{
        \label{fig:lfr-mu:subfig:rf_surgery_ari}
        \includegraphics[width=0.45\columnwidth]{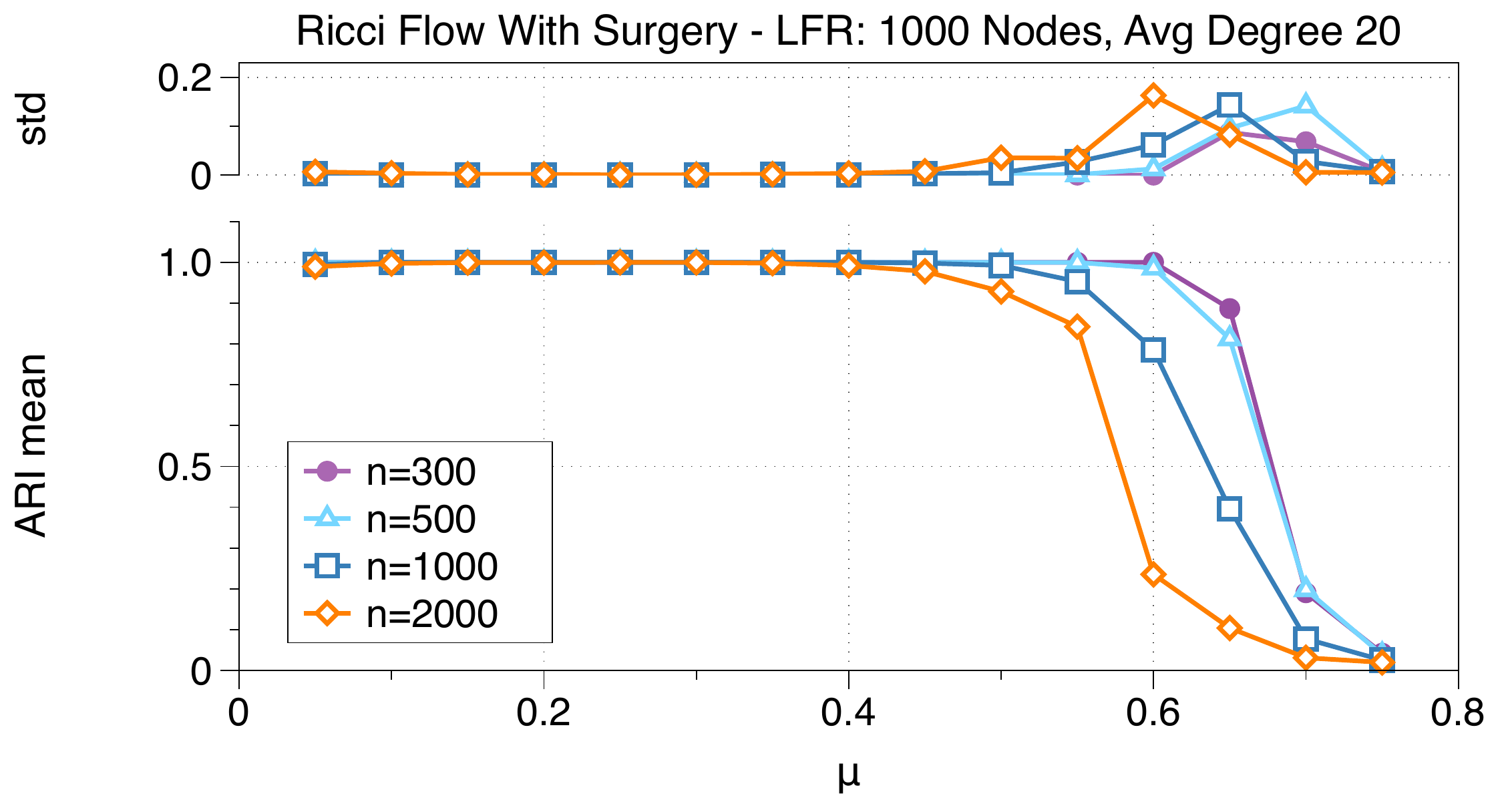}}
    \caption{ Accuracy evaluation of $50$ iterations of Ricci flow and $50$ iterations Ricci flow with surgery for every $5$ iterations. With LFR graph of $1000$ nodes and $\mu=0.6$, the surgery process boost the clustering accuracy ARI from $0.3$ to $0.8$.}
    \label{fig:surgery} %% label for entire figure
\end{figure}

For most of the networks that come without ground-truth community labels, we proposed to use modularity as an index to decide the final edge weight cutoff threshold to detect the communities. For a given graph, we first apply the graph with $20$ to $50$ iterations of discrete Ricci flow processes. (Iteration required varies by the complexity of the graph. For example, LFR with $1000$ nodes in Fig.~\ref{fig:surgery-iteration} takes around $50$ iterations to fully stabilized.) 
Then we remove edges from highest to lowest, and log the modularity as Fig.~\ref{fig:Gnet-modularity}. From the relationship of ARI and modularity we observed from Fig.6 in the main article, we suggest the cutoff threshold to be the point when modularity first hit the plateau of the curve. In Fig.~\ref{fig:Gnet-modularity} and Fig.~\ref{fig:Gnet}, we represent three different models of GNet from planar to scale-free to demonstrate the final clustering results with different cutoff thresholds based on modularity. By properly adjusting the cutoff threshold, the hierarchical community structure of the graph is also revealed. 

\begin{figure}[htbp]
    \centering
    \subfigure[GNet: $m=2, p=0.9$]{
        \label{fig:Gnet-modularity:subfig:m2p09}
        \includegraphics[width=0.31\columnwidth]{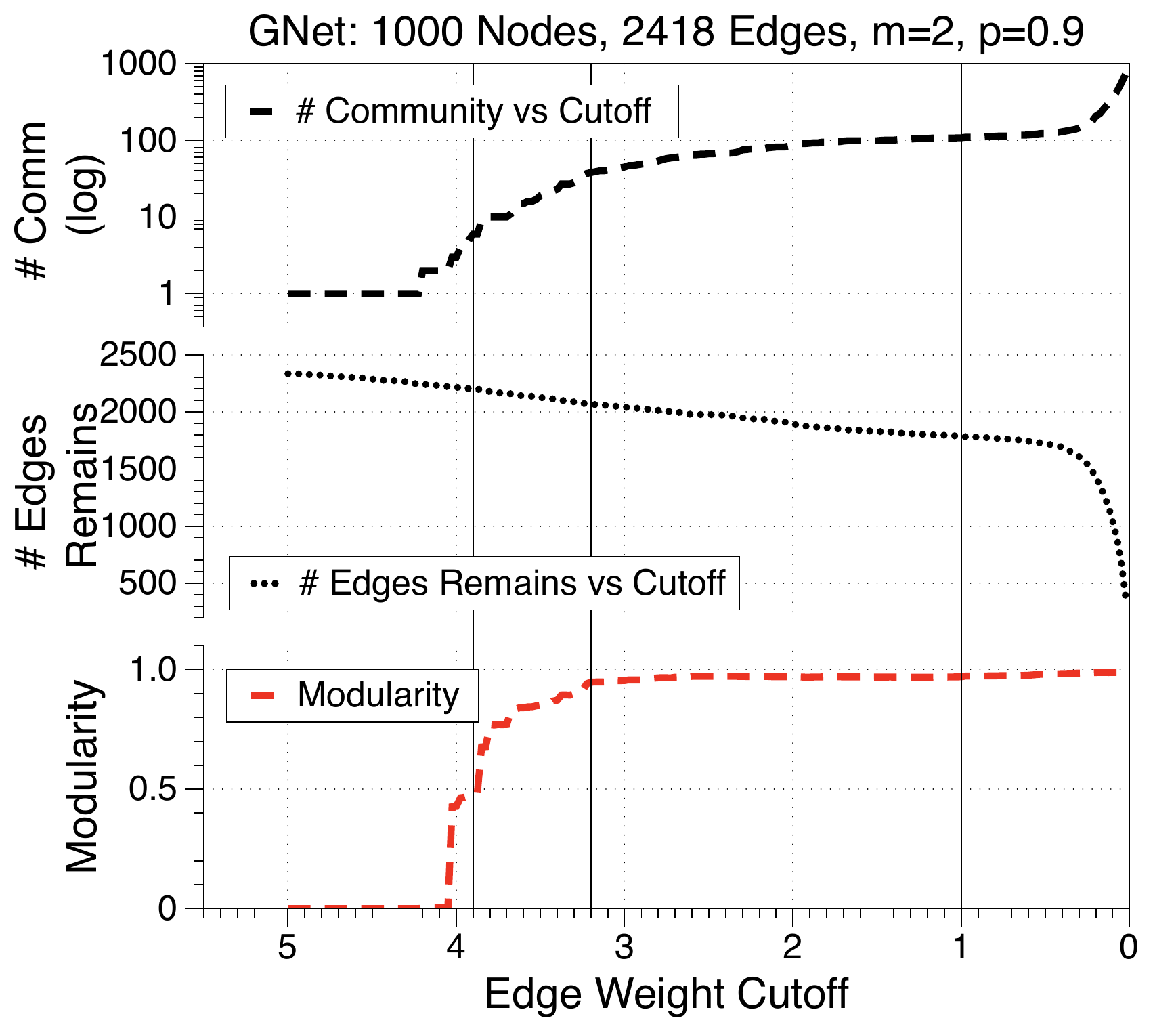}}
    \subfigure[GNet: $m=4, p=0.9$]{
        \label{fig:Gnet-modularity:subfig:m4p09}
        \includegraphics[width=0.31\columnwidth]{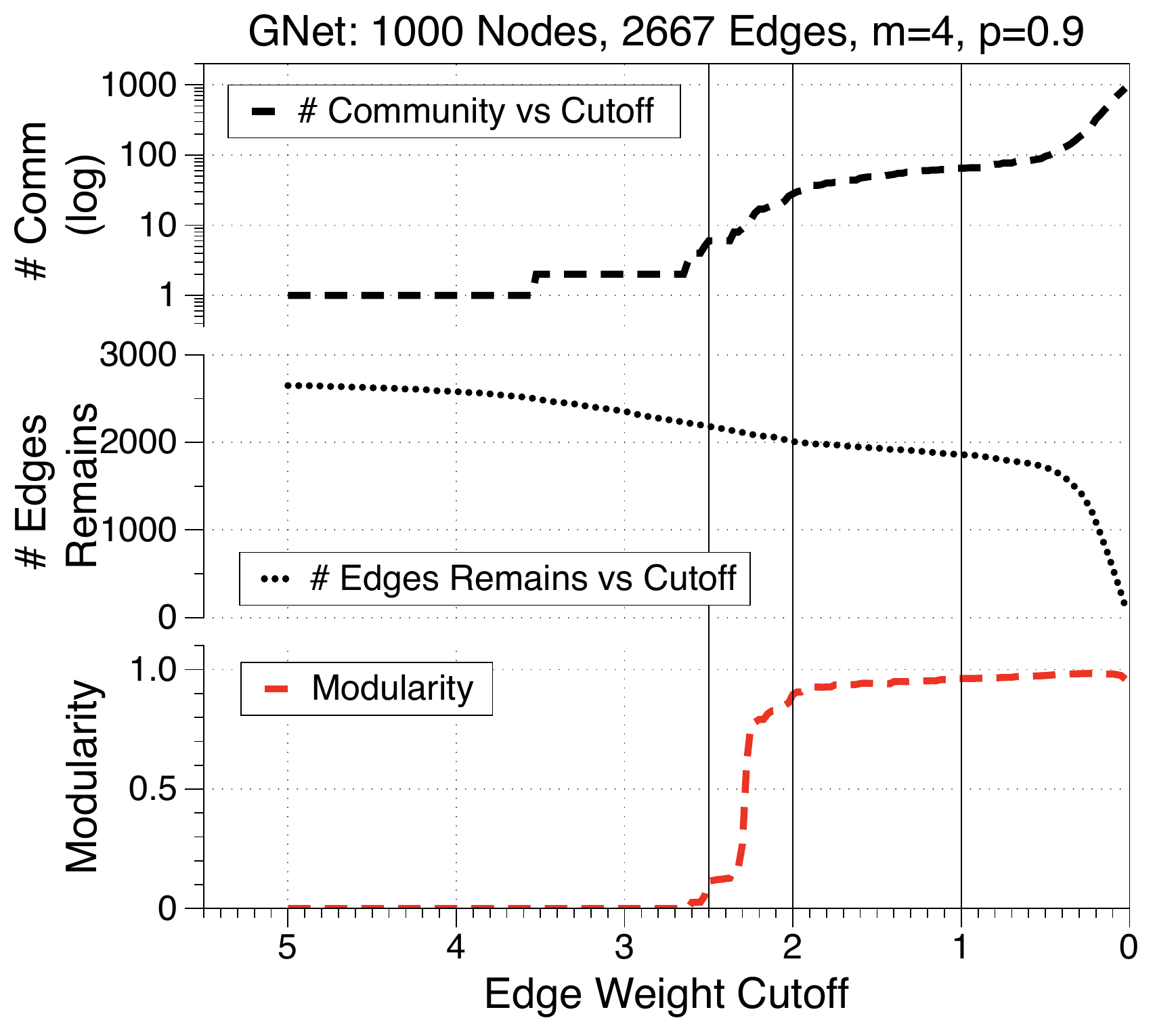}}
    \subfigure[GNet: $m=\infty, p=0.0$]{
        \label{fig:Gnet-modularity:subfig:m1p00}
        \includegraphics[width=0.31\columnwidth]{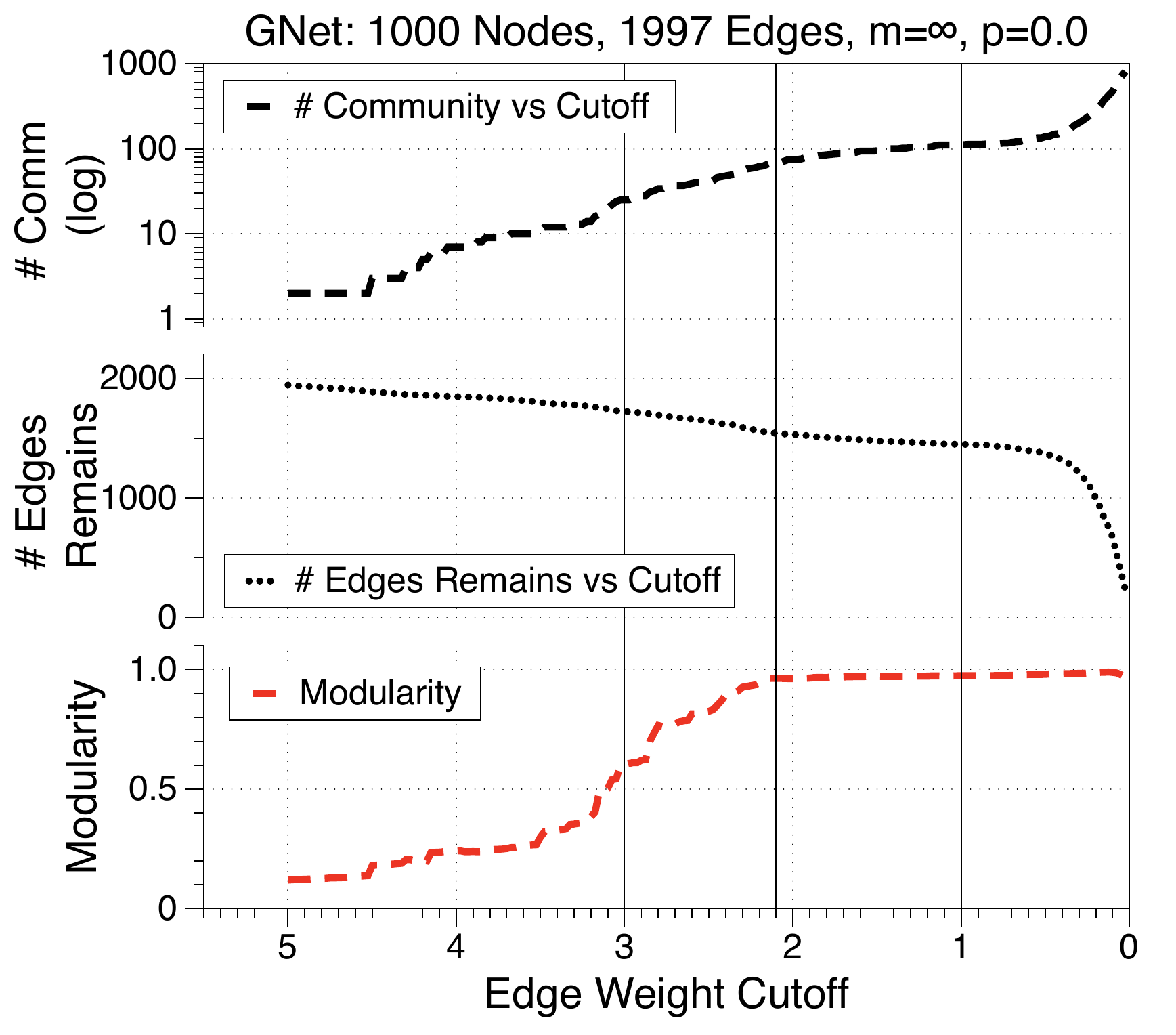}}
    \caption{Modularity of GNet with $1000$ nodes with different parameter setting. For each setting, we suggest the middle vertical line as cutoff threshold as it is the turning point for modularity curve. We also add two extra cutoff thresholds (left and right vertical line) for comparison. The result of communities detected with these three given cutoff thresholds are shown in Fig.~\ref{fig:Gnet}.}
    \label{fig:Gnet-modularity} %% label for entire figure
\end{figure}

\begin{figure}[htbp]
    \centering
    \includegraphics[width=1\columnwidth]{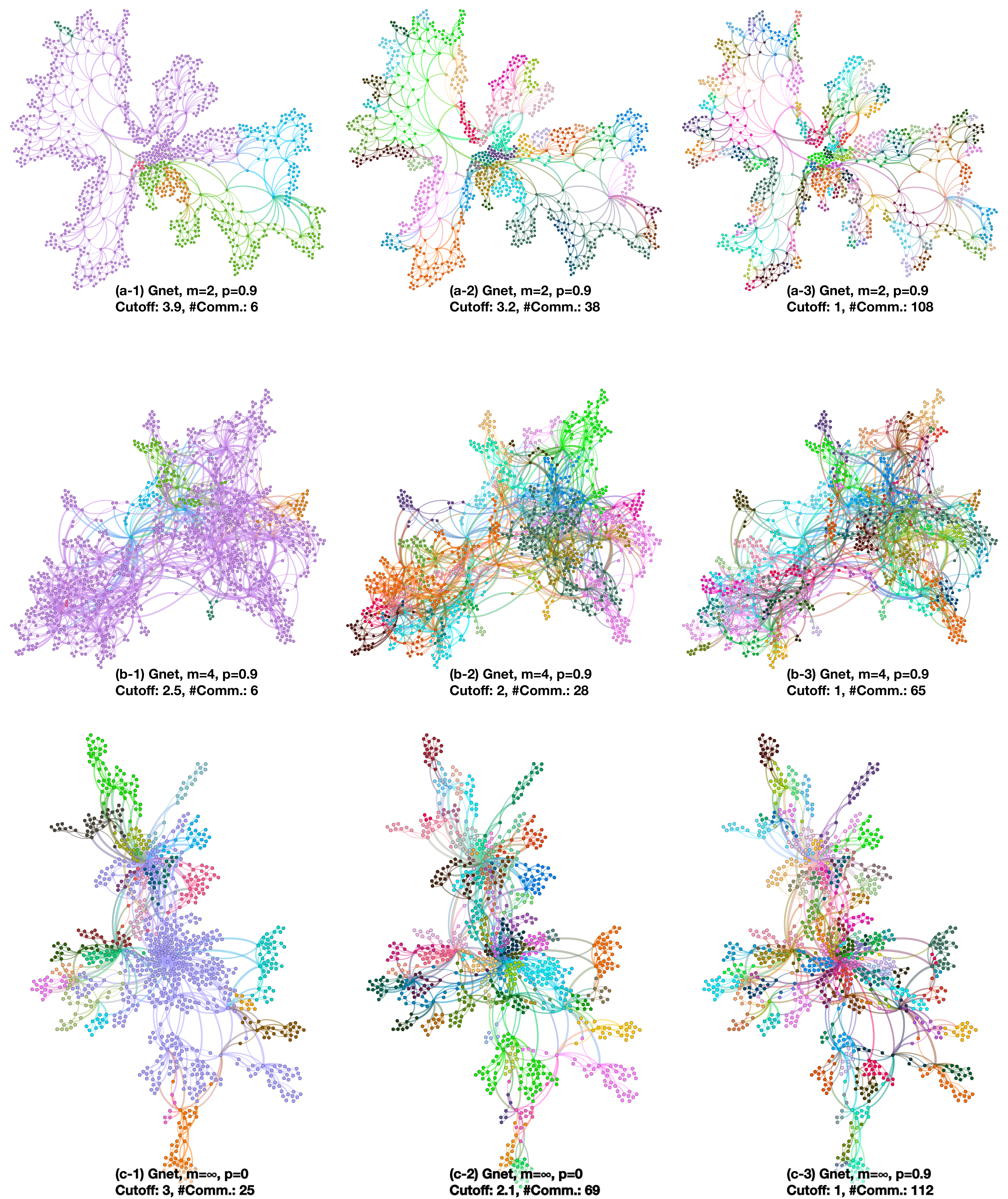}
    \caption{The emergent geometrical network model generates the following networks:\ref{fig:Gnet-modularity:subfig:m2p09} a network with random planar geometry for $m=2, p=0.9$; \ref{fig:Gnet-modularity:subfig:m4p09} a network with a broad degree distribution, small-world property, and finite spectral dimension for $m=4, p=0.9$; \ref{fig:Gnet-modularity:subfig:m1p00} and a scale-free network with power-law degree distribution for $m=\infty, p=0.9$.}
    \label{fig:Gnet} %% label for entire figure
\end{figure}
\section{Proof of Theorem 4.1}
\label{sec:convergence_proof}
We remark that we are not able to prove the similar result for other Ollivier-Ricci curvatures when $p>0$ though numerical results indicate it should be true.

We start by computing the Wasserstein distance of a metric on $G(a,b)$. In each community $C_i$ there is a specific node $u_i$ which connects to other communities. We call this node the gateway node and the rest of the nodes in $C_i$ the non-gateway nodes. There are three types of edges in the graph, edges connecting two communities (on two gateway nodes, such as $u_1u_2$ in Fig.~\ref{fig:G_a_b}), edges connecting a gateway node with a non-gateway node in the same community (such as $u_2i$ and $u_2 j$), and edges connecting two non-gateway nodes (such as $ij$).

\begin{figure}[htb]
    \centering
    \includegraphics[width=0.3\columnwidth]{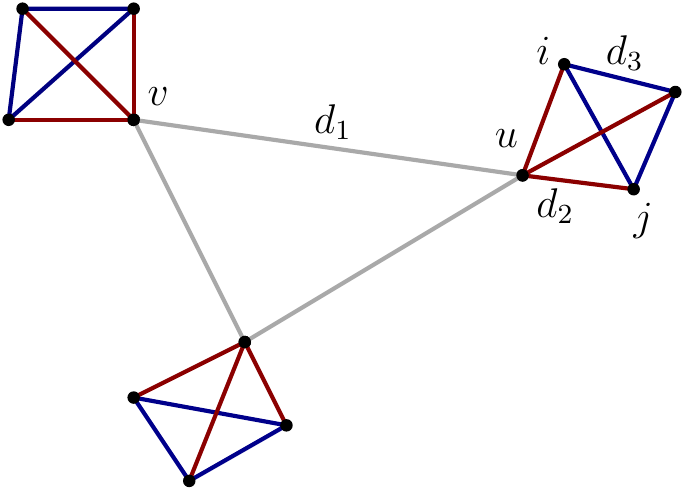}
    \caption{The graph $G(a,b)$ obtained from a complete graph on $b+1$ vertices by replacing each vertex by a complete graph of $a+1$ vertices. In this figure, $a=3$, $b=2$.} 
    \label{fig:G_a_b} 
\end{figure}

Since the initial metric has edge length one and the Ricci flow preserves the graph symmetry, there are only three different edge lengths at each iteration of the Ricci flow, corresponding to the three types of edges.   Suppose the edge lengths of the metric at the $n$-th iteration are $d_1$, $d_2$ and $d_3$ for the edges between communities, edges from a gateway node to a non-gateway node, and edges between two non-gateway nodes respectively, as shown in Fig.~\ref{fig:G_a_b}.  Let edge lengths of the $(n+1)$-th iteration be $D_1$, $D_2$ and $D_3$ which are the Wasserstein distances of the corresponding edges for the metric graph $(G(a,b), d)$ with respect to the probability measures $\{\mu_x| x \in V\}$.

\begin{lemma} The Wasserstein distances $D_1, D_2, D_3$ are given by

$ D_1 =\frac{a-1}{a+b}d_1+\frac{2a}{a+b}d_2,$

$ D_2=\frac{b}{a+b}d_1+\frac{ab-a-b}{a(a+b)}d_2+\frac{1}{a+b} d_3,$

$D_3=\frac{1}{a}d_3.$

 Furthermore, suppose $a>b \geq 2$ and $d_1\geq d_2 \geq d_3$. Then $D_1 \geq D_2 \geq D_3$. \end{lemma}
\begin{proof}
The easiest to compute is $D_3$.  The vertices adjacent to $i$ are the same as vertices adjacent to $j$. Furthermore, each vertex $x$ adjacent to
$i$  (or $j$) carries the same mass $1/a$. See Fig.~\ref{fig:G_a_b} (a). Thus the optimal transportation to move $\mu_i$ to $\mu_j$ is to transport the mass  $1/a$ at $j$ to $i$ along the edge $ij$ of distance $d_3$. Therefore,  $D_3=\frac{1}{a}d_3$ by definition.  See Fig. 2 (a).

\begin{figure}[htbp]
    \centering
    \subfigure[]{
        \label{fig:opt:subfig:a}
        \includegraphics[width=0.6\columnwidth]{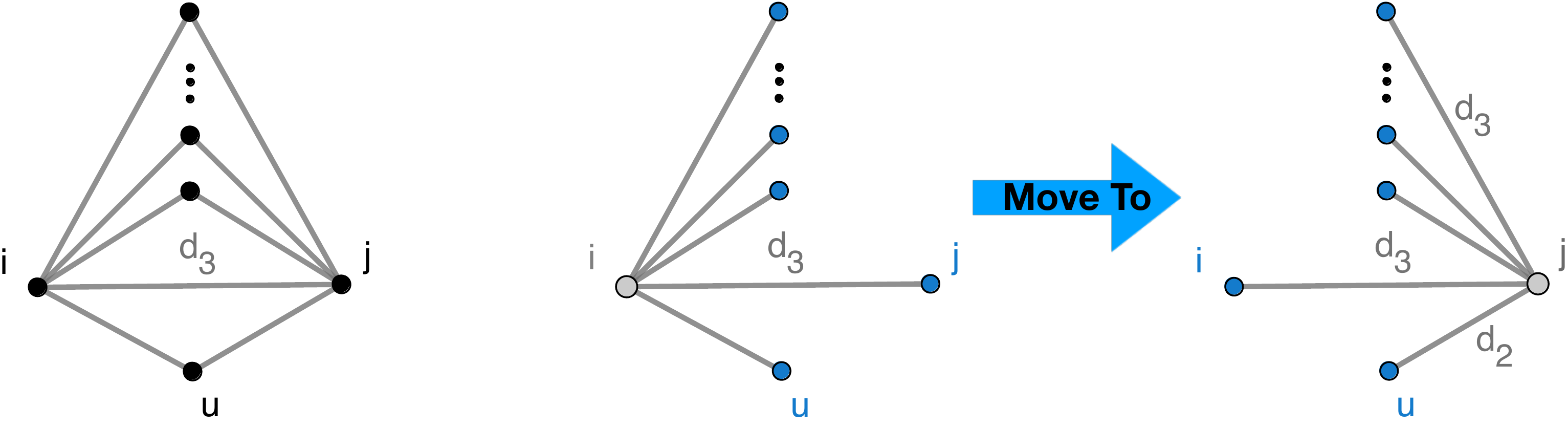}}
    \subfigure[]{
        \label{fig:opt:subfig:b}
        \includegraphics[width=0.6\columnwidth]{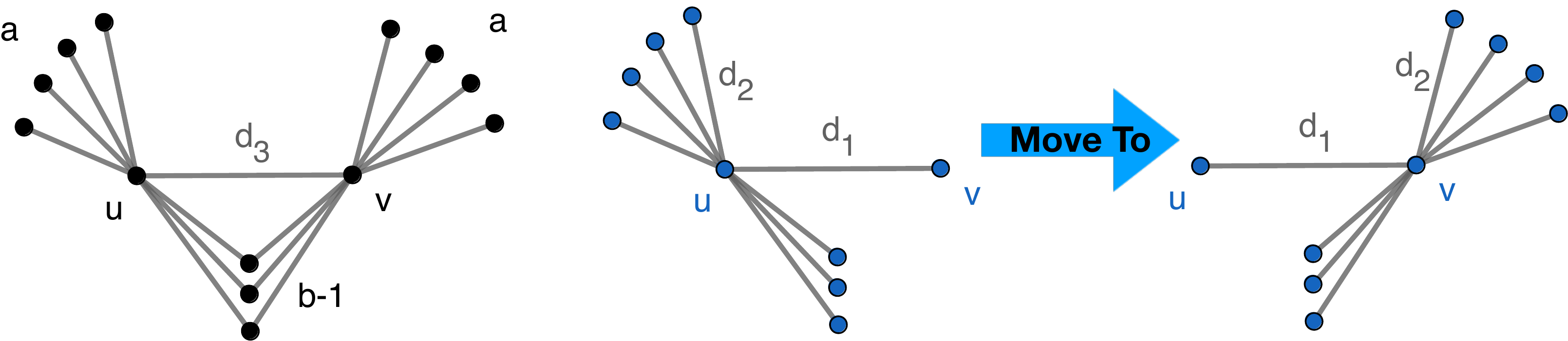}}
    \subfigure[]{
        \label{fig:opt:subfig:c}
        \includegraphics[width=0.6\columnwidth]{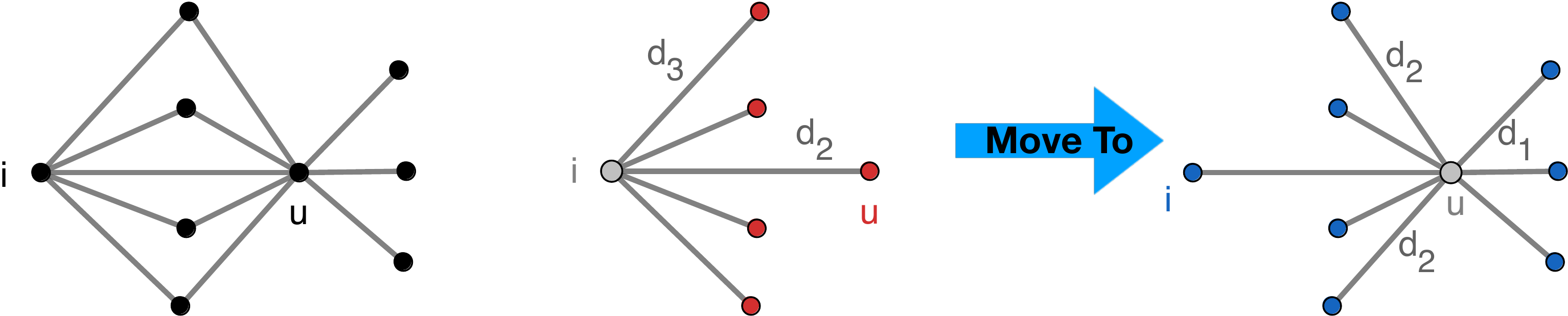}}
    \caption{Parts (a), (b) and (c) illustrate the optimal transportation to move the mass at vertex u to vertex v.}
    \label{fig:opt} %% label for entire figure
\end{figure}

Now let us compute the Wasserstein distance $D_1$ of moving the probability measures $\mu_u$ to $\mu_v$. Note that by definition the vertex degree $d_u=d_v=a+b$. Thus the probability measures $\mu_u$ and $\mu_v$ have mass $\frac{1}{a+b}$ at vertices adjacent to $u$ or $v$.
There are $b-1$ vertices which are adjacent to both $u$ and $v$ such that each vertex has mass $\frac{1}{a+b}$ in both $\mu_u$ and $\mu_v$. Therefore, there is no need to move them.  We only need to move the mass $\frac{1}{a+b}$ at each of  the rest $a$ many vertices $x$ adjacent to $u$ to those vertices $y$ adjacent to $v$ and need to move the mass at $v$ to the mass at $u$.  The best transportation plan goes as follows.
\begin{itemize}
\item Step 1. Move  the mass $\frac{1}{a+b}$ at vertex $x$ to $u$ along the edge $xu$. Since there are $a$ many such vertices $x$, the total cost is  $\frac{a}{a+b} d_2$  where $d_2$ is the distance from $x$ to $u$.

\item Step 2. Leave a mass of $\frac{1}{a+b}$ at $u$ and move the rest of mass $\frac{a-1}{a+b}$ from $u$  to $v$ along $uv$.  The total cost is $\frac{a-1}{a+b}d_1$.

\item Step 3.  Finally, move the mass $\frac{a}{a+b}$ at $v$ to vertices $y$ by equal distributions along edges $vy$ of distance $d_2$.
The total cost is $\frac{a}{a+b}d_2$.
\end{itemize}
Therefore, $D_1=\frac{a-1}{a+b}d_1+\frac{2a}{a+b}d_2$.

Finally, let us compute the
Wasserstein distance $D_2$ of moving the measure $\mu_i$ to $\mu_u$.  The mass of $\mu_i$ at a vertex $x$ adjacent to $i$ is $\frac{1}{a}$  and the mass of $\mu_u$ at vertices adjacent to $u$ is $\frac{1}{a+b}$.  The best transportation plan is the following.  Note that every vertex adjacent to $i$ is also adjacent to $u$.
\begin{itemize}
\item Step 1. Leave the mass $\frac{1}{a+b}$ at each vertex $x \neq u$ adjacent to $i$. The total leftover  mass at these vertices $x$ is  $\frac{(a-1)b}{a(a+b)}$.
Move the mass of  $\frac{1}{a+b}$ from the leftover mass at these $x$ to the vertex $i$ of distance $d_3$ from $x$. The total cost is $\frac{d_3}{a+b}$. After this move, there is a total mass of $\frac{(a-1)b}{a(a+b)}-\frac{1}{a+b}=\frac{ab-a-b}{a(a+b)}$ at these $x$.  Move them to vertex $u$ along edges of length $d_2$.
The total cost is
$\frac{ab-a-b}{a(a+b)}d_2$.

\item Step 2.  Finally, move a mass of $\frac{b}{a+b}$ at the vertex $u$ to vertices $y$ adjacent to $u$ such that $y$ is not adjacent to $i$. The total cost is $\frac{b d_1}{a+b}$.
\end{itemize}
Therefore, the cost of transportation is $D_2=\frac{d_3}{a+b}+\frac{ab-a-b}{a(a+b)} d_2+ \frac{bd_1}{a+b}$.

Now we come to prove the last statement.  Since $d_1 \geq d_2 \geq d_3$, we obtain
$D_2 \geq \frac{b}{a+b}d_1 \geq \frac{1}{a}d_1  \geq \frac{1}{a}d_3 =D_3$.
Also, $$D_1-D_2=\frac{a(a-b-1)}{a(a+b)} d_1+\frac{2a^2-ab+a+b}{a(a+b)} d_2 -\frac{a}{a(a+b)} d_3$$ $$
\geq \frac{d_2}{a(a+b)}[ a(a-b-1)+2a^2-ab+a+b-a]$$ $$
=\frac{d_2}{a(a+b)}[a(a-b-1) +a^2-ab +a^2+b] \geq 0.$$
\end{proof}

Let $A$ be the $3 \times 3$ matrix
$$
\begin{bmatrix}
    \frac{a-1}{a+b}    & \frac{2a}{a+b} &  0 \\
    \frac{b}{a+b}      & \frac{ab-a-b}{a(a+b)}&    \frac{1}{a+b}\\
      0   & 0  & \frac{1}{a}
\end{bmatrix}
$$
and $W_n=[w_{n1}, w_{n2}, w_{n3}]^t$ be the $3 \times 1$ column vector given by $W_{n+1}=AW_{n}$ with $W_0=[1,1,1]^t$.  Then $w_{ni}$ are the edge lengths of the graph $G(a,b)$ after $n$-th iteration of the Ricci flow.  Our goal is the understand the asymptotic behavior of $W_n$.

Using Maple calculation, we conclude the following lemma.

\begin{lemma} Suppose $a > b \geq 2$, then there are three real eigenvalues $\lambda_1 > \lambda_2= \frac{1}{a} \geq 0 > \lambda_3$ of the matrix $A$.
The largest eigenvalue $\lambda_1 > \lambda_3$. Furthermore, an eigenvector $w_1$ associated to $\lambda_1$ is of the form $[1, k, 0]^t$ where $k$ is in the open interval $(0,1)$.
\end{lemma}

By this lemma, we conclude the proof of the theorem as follows. Let $w_2$ and $w_3$ be the eigenvectors associated to $\lambda_1$ and $\lambda_3$ of $A$. The $\lambda_2=\frac{1}{a}$ eigenvector $w_2 =[0,0,1]^t$.  Since $\lambda_1, \lambda_2, \lambda_3$ are distinct, the eigenvectors $w_1, w_2, w_3$ are linearly independent in $\R^3$.

Write the vector $W_0=[1,1,1]^t=a_1w_1+a_2w_2+a_3w_3$. Then using $w_1=[1,k,0]^t$ and $\lambda_1 > |\lambda_2|, |\lambda_3|$, we obtain
$$W_n = a_1\lambda_1^n w_1+a_2\lambda_2^n w_2+ a_3\lambda_3^n w_3 =
\begin{bmatrix}  a_1 \lambda_1^n +o(\lambda_1^n) \\ k a_1 \lambda_1^n +o(\lambda_1^n)  \\ (\frac{1}{a})^n \end{bmatrix}.$$
Here $o(\lambda_1^n)$ stands for an expression such that $\lim_{n \to \infty} o(\lambda_1^n)/\lambda_1^n=0$.  As a conclusion, we see that the distance at the edge $uv$ grows at the rate of $\lambda_1^n$, the distance at the edges $ui$ and $ij$ grows at rate $o(\lambda_1^n)$. In fact, the distance the edge $ij$ grows at the rate of $\frac{1}{a}<1$ and shrinks to zero exponentially fast.

\bibliography{jiepub,community,yuylin,ccni}

\end{document}